\documentclass[12pt]{amsart}
\usepackage[normalem]{ulem}
\usepackage{amssymb}
\usepackage{amsmath,amsthm,amsfonts}
\usepackage[applemac]{inputenc}
\usepackage{graphicx}
\usepackage{pdfsync}
\usepackage{fancyhdr}
\usepackage{geometry}
\usepackage{caption}
\usepackage{bm}
\usepackage{color}
\usepackage{xcolor}
\usepackage{tikz}
\usepackage{float}
\usepackage{mathtools}
\usepackage{tensor}
\usepackage{subcaption}
\usepackage[breaklinks,colorlinks,backref]{hyperref}
\hypersetup{
    colorlinks, %set true if you want colored links
    linktoc=all, %set to all if you want both sections and subsections linked
    linkcolor=red,  %choose some color if you want links to stand out
  citecolor=hyptxt,
  urlcolor=blue}
  \hypersetup{
  citebordercolor=,
  filebordercolor=green,
  linkbordercolor=blue
}
\definecolor{hyptxt}{rgb}{0.7, 0.4, 0.9}
\usepackage[full]{textcomp} % to get the right copyright, etc.
     \usepackage{fbb} % osf for text, lining for math
     \usepackage[scaled=.95]{cabin}
     \usepackage[varqu,varl]{inconsolata}% typewriter
     \usepackage[libertine,bigdelims,vvarbb]{newtxmath}
     \usepackage[cal=boondoxo]{mathalfa}% less slanted than STIX
     \usepackage[T1]{fontenc}
\useosf
\newtheorem{proper}{Conditions}[section]
\newtheorem{defi}{Definition}[section]
\newtheorem{prop}{Proposition}[section]

\newtheorem{theo}{Theorem}[section]

\newcommand{\ket}[1]{|\kern.3ex#1\kern.3ex\rangle}
\newcommand{\bra}[1]{\langle\kern.3ex #1 \kern.3ex|}
\newcommand{\scalar}[2]{\langle\kern.3ex #1 \kern.3ex|\kern.3ex#2\kern.3ex\rangle}

\DeclareMathOperator{\okr}{{\stackrel{{\scriptscriptstyle{\mathsf{def}}}}{=}}}
\DeclareMathOperator{\E}{e}
\def\Da{\mbox{\Large \textit{a}}}

\def\R{\mathbb{R}}
\def\N{\mathbb{N}}
\def\C{\mathbb{C}}
\def\Z {\mathbb{Z}}

\def\lg{\langle }
\def\rg{\rangle }
\def\deq{\stackrel{\mathrm{def}}{=}}

\def\ud{\mathrm{d}}

\def\ii{\mathrm{i}}
\def\S{\mathbb{S}^{1}}

\def\Da{\mbox{\Large \textit{a}}}
\def\De{\mbox{\Large \textit{e}}}

\definecolor{hervecolor}{rgb}{0.8,0,0.7}

\numberwithin{equation}{section}

\date{\today}

\begin{document}
%--------------------------------------------------------------------------------------
\title{Quantum localisation on the circle}
\author{Rodrigo Fresneda$^{a}$, Jean Pierre Gazeau${}^{b,c}$, and  Diego Noguera${}^c$}
\address{$^{a}$ Federal University of ABC (UFABC)
Santo Andr\'e - SP - Brasil . CEP 09210-580}
\address{${}^b$ APC, UMR 7164,
Univ Paris  Diderot,  Sorbonne Paris-Cit\'e
75205 Paris, France}
\address{${}^c$ Centro Brasileiro de Pesquisas Fisicas
Rua Xavier Sigaud 150, 22290-180 - Rio de Janeiro, RJ, Brazil}
\email{rodrigo.fresneda@ufabc.edu.br, gazeau@apc.in2p3.fr, diegomac@cbpf.br}
\begin{abstract}
Covariant integral quantisation using coherent states for semidirect product groups is studied and applied to the motion of a particle on the circle. In the present case the group is the Euclidean group E$(2)$. We implement the quantisation of the basic classical observables, particularly the $2\pi$-periodic discontinuous angle function and the angular momentum, and compute their corresponding lower symbols. An important part of our study is devoted to the angle operator  given by our procedure, its spectrum and lower symbol, its commutator with the quantum angular momentum, and the resulting Heisenberg inequality.  Comparison with other approaches to the long-standing question of the quantum angle is discussed. 
 \end{abstract}

\maketitle
\tableofcontents
%---------------------------------------------------------------------------------------
\section{Introduction}
The simple pendulum is a familiar example of an elementary one-dimensional conservative mechanical system.  The model is  
 a point particle of mass $m$ moving on a (portion of) a vertical circle of radius $l$ and subject to the  potential $V(\theta) = - mg\cos\theta$, where $\theta$  is the position angle measured from the lowest position. This problem is pedagogically  interesting, since it represents a link between two extreme situations, namely, the pure rotor (for $g = 0$) and the harmonic oscillator for small $\theta$.
The natural  canonical coordinates are $(p_{\theta}, \theta)$ where $p_{\theta} = m l^2\dot \theta $ is the angular momentum in action units.
 In terms of these variables, the Hamiltonian reads
\begin{equation}
\label{hamsimpend}
H = \frac{p_{\theta}^2}{2ml^2} - mgl \cos\theta\, ,   \quad -\pi\leq \theta < \pi\, , 
\end{equation}
and is conserved: $H=E$.  According to the values of $E$, we distinguish between 2 types of motion, namely 
 rotation at large enough energy  $E> mgl$, libration at small enough energy $E< mgl$, with a bifurcation  separatrix at $E=mgl$. 
  
When we wish to establish the quantum version of this classical model, we face the difficulty of properly defining a localisation operator on the circle, whereas such an object exists unambiguously for  the quantum model of the motion on the line. Indeed, supposing that a $2\pi$ periodic wave functions $\psi(\theta)$ exists on the circle, we cannot introduce an angle operator $\hat \theta$ as the multiplication operator $(\hat \theta\psi)(\theta) = \theta \psi(\theta)$ without breaking  the periodicity, except if the factor $\theta$ stands for the $2\pi$-periodic discontinuous angle function, i.e., 
\begin{equation}
\label{angleop1}
(\hat \theta\psi)(\theta) := \left(\theta -2\pi\,\left\lfloor \frac{\theta}{2\pi}\right\rfloor \right)\psi(\theta)\,,
\end{equation} 
as given for instance in \cite{levyleblond76} (see also \cite{kowal02}), and where $\lfloor\cdot\rfloor$ stands for the floor function.  On a more mathematical level, if we require  $\hat \theta$ to be a self-adjoint multiplication operator with spectrum supported by the period interval $[0,2\pi)$, on which are defined these $2\pi$ periodic wave functions, it is well known that the canonical commutation rule $\left[\hat \theta,  \hat p_{\theta}\right]= \ii \hbar I$ cannot hold with a  self-adjoint quantum angular momentum $\hat p_{\theta}=-\ii\hbar\frac{\partial}{\partial \theta}$. In consequence, it is necessary to revisit the  quantum localisation on the circle and its related Heisenberg inequality $\Delta \hat \theta\, \Delta  \hat p_{\theta} \geq \mathrm{glb}$. Most of the approaches and subsequent discussions rest upon the replacement of a hypothetical angle operator with the quantum version of a smooth periodic function of the classical angle, at the cost of the loss of satisfying localisation properties. 

Perhaps an even simpler physical system where the circular nature of a coordinate manifests itself is the torsion pendulum. Classically, it is an oscillator whose restoration
force is torque. Consider a disc suspended by a wire, free to rotate
along its axis. As the disc rotates by an angle $\theta$ of twist from its equilibrium position in radians, the wire
resists such deformation by developing a restoring torque $\tau=-k\theta$,
where $k$ is the torque constant of the wire. The rotation equation
of motion, neglecting the moment of inertia of the wire, is then the
simple harmonic equation $I\ddot{\theta}=-k\theta$, where $I$ is
the moment of inertia of the disk. The Hamiltonian has the form $H_{tp}=\frac{1}{2I}p^{2}+\frac{k}{2}\theta^{2}$,
and the oscillator frequency is $\mu=\left(k/I\right)^{1/2}$. It
is convenient to write the Hamiltonian in terms of the dimensionless
quantities $\mathfrak{p}=\left(I\mu\hbar\right)^{-\frac{1}{2}}p$, $\mathfrak{q}=\left(I\mu/\hbar\right)^{\frac{1}{2}}\theta$,
$H=\left(\mu\hbar\right)^{-1}H_{tp}$, so that $H=\frac{1}{2}\left(\mathfrak{p}^{2}+\mathfrak{q}^{2}\right)\,.$
Now, although the coordinate $\mathfrak{q}$ is supposedly a periodic angle coordinate, it assumes its value in a subinterval of $(-\pi, \pi)$ for obvious reasons. Here too, a consistent localisation is required when one deals with the (hypothetical) quantum version of this model. 

%More bibliographical details are given below. 

Let us now present  a survey of the huge literature on this problematic angle operator conjugate to the quantum angular momentum, or on its parent phase operator conjugate to the number operator.  Since Dirac's proposal \cite{dirac27}   of a phase operator in 1927 (see also \cite{dirac58}) and his mention of a canonical commutator between phase and number operators, there has been a wealth of works dedicated to the description of the phase and angle operators. An early review focused on the commutation relations between angle and angular momentum and between phase and number operators is the celebrated \cite{carruthers1968}.  For a review in the context of electromagnetic theory, see \cite{lynch1995}.  In the following we do not attempt to exhaust all the literature, but we only highlight some of its landmarks. In the wake of Dirac's proposal, early works were fixed on the goal of attaining canonical commutation relations  which reproduce the classical Poisson brackets, in analogy to position and momentum. It was soon realized that this was not possible, due to angular momentum operator domain issues \cite{judge1963}. The non-canonical commutation relation
\begin{equation}
\label{comrel1}
\left[\hat{\theta},\hat p_{\theta}\right]=\ii \hbar \left[1-2\pi \sum_n \delta(\theta -2n\pi)\right]\, , 
\end{equation}
a trivial consequence of \eqref{angleop1}, was used to derive modified uncertainty relations \cite{lewis1963}, which were deemed incorrect  for not taking into account the domain of the commutator. Then a new uncertainty relation was proposed  for self-adjoint operators $A$ and $B$ based on a symmetric bilinear form $\Phi_{A,B}(f,g)$ which reproduces the matrix elements of the non-canonical commutator  on eigenfunctions of the angular momentum operator. However, in order to correct the lack of rotation symmetry which follows from the uncertainty relations $\Delta A\, \Delta B \geqslant 1/2 \left|\Phi_{A,B}(f,g) \right|$, yet another uncertainty relation was put forward by \cite{kraus1965}. According to \cite{louisell63}, the phase has no direct physical meaning in terms of measurement, so periodic functions of the phase such as sine and cosine were chosen as the relevant quantities in deriving commutation relations with the number operator. The idea of using periodic functions in place of the angle variable has been pursued in many works where a formal  operator algebra is used to derive uncertainty relations \cite{susskind1964,lerner1970}. The trick is to find  operators $C$ (for cosine) and $S$  (for sine) such that  $C+iS$ is unitary, thus defining a self-adjoint phase operator. This approach was then put on rigorous grounds, where a self-adjoint phase operator was defined which has canonical commutation relations with the number operator in \cite{garrison1970,alimow1979,galindo1984} and general properties of phase operators are found in \cite{ifantis1971}. For other generalizations, one recalls the circular operators in  \cite{mlak1992}, for which phase operators are particular cases, and the enlarged Hilbert space for defining negative integral values of the number operator \cite{newton1980} . A rigorous treatment based on the canonical factorization theorem of the phase operator in the context of quantum electrodynamics is given in \cite{volkin1973}. In \cite{rocca1973} a no-go theorem is proved about the nonexistence of a phase operator along the lines of the previous works for systems with finite degrees of freedom. 

In \cite{levyleblond76} the author proposes that instead of the usual restriction of observables to hermitian operators (or rather, self-adjoint operators), one consider non-hermitian operators (or rather, non-self-adjoint operators) whose eigenstates still provide a resolution of the identity and the necessary probabilistic interpretation of the Hilbert space formalism. In such a manner, one is able to obtain meaningful commutation relations between the angular coordinate and the angular momentum, or between the phase operator and the number operator. For instance, in the case of the angle operator, the author proposes that in place of the discontinuous angle operator, one should consider the corresponding unitary operator $\exp i\varphi$, which has well-defined commutation relations with the angular momentum operator. As a result, one has the familiar Heisenberg inequality proposed by \cite{carruthers1968}.

A work by Royer \cite{royer96} related to these questions also deserves to be mentioned.  A popular approach \cite{barnett1988,popov1992} is one in which the phase and number operators are defined in an $n$-dimensional space, so all computations are performed prior to taking the $n\rightarrow \infty$ limit (for a review, see \cite{barnett2007}).  For works more specifically oriented towards quantum measurement, see \cite{busch01,bukiuwer16}. 

Most of these works address the question of the validity of commutation relations between the phase (resp. angle) operator, defined in a particular way, e.g., by means of a smooth function,   and the number operator (resp. angular momentum) (see also \cite{galapon02} for a clear mathematical analysis).  One of our aims in the present work is to build, from the classical angle function acceptable angle operators through a consistent and manageable quantisation procedure.  We recall that the standard ($\sim$ canonical) quantisation is based on the replacement of the classical conjugate pair $(q,p)\in \R^2$, with $\left\lbrace q , p \right\rbrace = 1$, by its quantum counterpart $(Q,P)$ made of two essentially self-adjoint operators having continuous spectrum $\R$ and such that $[ Q , P ] = i\hbar I$. As a result, the quantisation of a classical observable is the (not well-defined) map  $f(q,p) \mapsto \mathrm{Sym} f(Q,P)$, where the symbol $\mathrm{Sym}$ stands for symmetrisation, which maps real functions to symmetric operators.  Due to the pragmatic stance of the procedure,  canonical quantisation  is commonly accepted in view of its numerous experimental validations since the emergence  of quantum physics.  Now, when one wants to implement the method in dealing with  geometries  other than simple Euclidean spaces, particularly when one is concerned with impenetrable barriers, or when one wants to quantise singular functions, one may be faced with serious mathematical problems. This is precisely  the case we are considering in this article, namely the discontinuous angle (or phase function $\text{arctan}\, p / q$ in the above case), for which canonical quantisation is clearly unsuited.  

In the present work we revisit the problem of the quantum angle through coherent state (CS)  quantisation, which is a particular (and better manageable) method belonging to (covariant) integral quantisation \cite{bergaz13,aagbook14,gabafre14}.    Various families of coherent states have  already been used for this purpose, as the standard or the so-called circle coherent states or even more general versions like the ones in \cite{gazszaf16,argaho12}. CS quantisation has also been applied in the finite-dimensional Hilbertian framework in \cite{pedro2007},  where infinite-dimensional  limits are taken of  mean-values of physical quantities in order to obtain the usual commutation relations between phase and number operators. The essential ingredient of CS quantisation or the more general integral quantisation is the resolution of the identity provided by a (positive) operator-valued measure. Here, our approach is group theoretical, based on unitary irreducible representations of the Euclidean group E$(2)= \R^2 \rtimes$SO$(2)$, and it is strongly influenced by the seminal paper by De Bi\`evre  \cite{debievre89} and chapter 9 of the book \cite{aagbook14}. Related group theoretical approaches are found in \cite{bowo74,isham84,niatchwo98,kastrup06}.

Let us give an overview of covariant integral quantisation of functions (or distributions when allowed by context) defined  on a homogeneous space $X$, viewed as the left coset manifold  $X\sim G/H$,  for the action of a Lie group $G$, where the closed subgroup $H$ is the stabilizer of some point of $X$. The case when $X$ is a symplectic manifold (e.g., a co-adjoint orbit of $G$) is of particular interest since it may be viewed as the phase space for the dynamics of some system. Suppose that $X$ is   equipped with a quasi-invariant measure $\nu$, that is,
${\ud}\nu(g^{-1}x)=\lambda(g,x){\ud}\nu(x)\;\;(\forall g\in G)$, with $\lambda(g,x)$  obeying the cocycle condition
$\lambda(g_1 g_2,x)= \lambda(g_1,x)\,\lambda(g_2, g_1^{-1}x)$.
  For a global Borel section $\sigma:X\rightarrow G$
of the group, let $\nu_{\sigma}$ be the unique quasi-invariant measure
defined by 
\begin{equation} \label{}
{\ud}\nu_{\sigma}(x)=\lambda(\sigma(x),x)\,{\ud}\nu(x)\,,
\end{equation} 

Let $U$ be a UIR which is square integrable $\mathrm{mod}(H,\sigma)$ with 
an \textit{admissible} density operator $\rho$, i.e., $\rho \geq 0$, $\mathrm{Tr} (\rho)=1$, and
\begin{equation}
\label{admissrho}
 c_\rho:=\int_{X}\,\mathrm{tr}\left(\rho\,\rho_{\sigma}(x)\right)\,{\ud}\nu_{\sigma}(x)<\infty  \quad \mbox{with} \quad
\rho_{\sigma}(x):=U(\sigma(x))\rho U(\sigma(x))^{\dag}\,. 
\end{equation}
Then \textit{square-integrability $\mathrm{mod}(H,\sigma)$} entails that we
have the resolution of the identity 
\begin{equation}\label{resident}
I= \frac{1}{c_\rho}\int_{X}\rho_{\sigma}(x)\,{\ud}\nu_{\sigma}(x)\, . 
\end{equation} 
The latter allows us to implement integral quantisation of functions (with possible extension to distributions) on $X$ , which is defined as the linear map
\begin{equation}\label{quantres}
 f\mapsto A^{\sigma}_{f}=\frac{1}{c_\rho}\int_{X}\; f(x)\,\rho_{\sigma}(x)\,{\ud}\nu_{\sigma}(x)\, , 
\end{equation} 
where we have labeled the dependence on section $\sigma$ . 

Covariance holds in the following sense (see Chapter $11$ in \cite{aagbook14} ). 
Consider  the sections $\sigma_{g}: X \rightarrow G, \;\; g \in G$, which are 
covariant translates of $\sigma$ under $g$:
\begin{equation}
 \sigma_{g}(x) = g\sigma (g^{-1}x) = \sigma (x)h(g, g^{-1}x)\, .
\end{equation}
where  the cocycle $h(g,x)$ belongs to $H$.
%\begin{equation}
%g\sigma (x) = \sigma (gx)h(g,x) \quad \mbox{with} \quad
%h(g,x)  \in H\, .
%\label{cocyclecond2}
%\end{equation}
Given 
 ${\ud}\nu_{\sigma_g}(x) := \lambda (\sigma_{g} (x), x)\,{\ud}\nu(x)$
 define 
  $$\rho_{\sigma_g}(x) = 
U(\sigma_{g}(x))\rho U(\sigma_{g}(x))^{\dag}\, .$$
 For $U$ square integrable  $\mathrm{mod}(H, \sigma )$, 
the general covariance property of the integral quantisation  $f\mapsto A^{\sigma}_f$ reads 
\begin{equation}
\label{quantizcovmodH2}
 U(g) A^{\sigma}_f U(g)^{\dag} = A^{\sigma_g}_{\mathcal{U}_l(g)f}\, , \quad A^{\sigma_g}_f:= \frac{1}{c_\rho}\int_X \rho_{\sigma_g}(x) f(x) \ud \nu_{\sigma_g}(x)\,., 
\end{equation} 
with $\mathcal{U}_l(g)f(x)= f\left(g^{-1}x\right)$.
Similar results are obtained by replacing $\rho$ by  more general bounded operators ${\sf M}$, provided integrability and weak convergence hold in the above expressions. When  $\rho$ is a rank-one density operator, i.e. $\rho = |\eta\rg\lg\eta|$, one says that $\eta \in \mathcal{H}$ is admissible, and we are working with CS quantisation, where the CS's are defined as 
$|\eta_x\rg := |U(\sigma_g(x))\eta\rg$. In this restricted context, $\eta$ is also called fiducial vector.

%When the representation $U$  itself is induced from a UIR $U_K$of a subgroup  $K\subset G$ acting on a Hilbert space $\mathcal{K}$,  it is constructed as $G\ni g\mapsto \widetilde{U}(g)\in\text{Aut}(\mathcal{K}\otimes L^{2}(G/H,d\nu))$. With some section $\sigma:G/H\rightarrow G$, we define the family of CS's as $\eta_{x}=U(\sigma(x))\eta$ for $\eta$ being an admissible  vector in $\mathcal{K}\otimes L^{2}(G/H,d\nu)$ and $x\in G/H$, such that the resolution of the identity $\int_{G/H}d\nu(x)\,\ket{\eta_{x}}\bra{\eta_{x}}=I$ holds in $\mathfrak{K}\otimes L^{2}(G/H,d\nu)$. 

The organisation of the paper is as follows.  In Section \ref{sec:gralset} we specify the above formalism  to Lie groups $G$ which are semi-direct products  of the type $G=V\rtimes S$, where $V$ is an $n$-dimensional vector space and $S$ is a subgroup of $GL(V)$. We present some important isomorphisms in order to construct the phase space of a physical system  having as a configuration space a certain  coset manifold $G/H$.  We also present  the notion of induced representations of the group $G$ using a representation of the subgroup $H$ of $G$. Using this representation of $G$ we construct a family of CS's  and explain how to perform covariant integral quantisation of functions on $G/H$.
In Section \ref{sec:cohe} we apply the above formalism to one of the simplest cases, namely the Euclidean group $\text{E}(2)$ which is the semidirect product  $E(2)=\R^2\rtimes \mathrm{SO}(2)$ and we introduce coherent states for $\text{E}(2)$ along the lines described above. In Section \ref{sec:quant} the corresponding covariant CS  quantisation is implemented. In this case, the $G$-coset $X=G/H= \left(\R^2\rtimes \mathrm{SO}(2)\right)/\R$  is represented by the cylinder $X = \R\times \mathbb{S}^1= \{(p,q)\,, \, p\in \R\, , \, q\in [0,2\pi)\}$. The configuration manifold is the unit circle on which the motion of the particle takes place, and the velocity is parametrized by $p$.  CS quantisation maps  functions $f(p,q)$  to operators $A_{f}$ (here we drop the $\sigma$ dependence for the sake of simplicity) in the Hilbert space $\mathcal{H}$ carrying the group representation $U$ of $G$ . When $f$ is real, i.e. when it is viewed as a classical observable, we expect that $A_{f}$ be self-adjoint, or at least symmetric.  We study the cases  where the function $f(p,q)=u(q)$ does not depend on $p$ and leads to a multiplication operator, the elementary example $u(q)=e^{\ii nq}$,  the quantum angular momentum issued from $f(p,q)=p$, the kinetic energy $f(p,q)=p^2$, as well as products of the type $pu(q)$ or $p^{2}u(q)$, in order to cover the majority of the interesting Hamiltonians in quantum mechanics.  In Section \ref{sec:lower} a family of probability distributions is constructed from the CS. They provide a semi-classical portrait $\check{f}(p,q)$ associated to the operator $A_{f}$. Explicit formulas are given for $\check{u}(q)$ and $\check p$.  Section \ref{sec:spectral} is devoted to the study of the  angle operator  resulting from the quantisation of the $2\pi$-periodic discontinuous angle function, particularly its  spectrum as a bounded self-adjoint multiplication operator.  The study is illustrated analytically and numerically with the use of a  particular family of smooth fiducial vectors $\eta$. Section \ref{heisenberg} is devoted to the analytic and numerical study of the commutation relation  between  quantum angle and quantum angular momentum and the resulting uncertainty relation or Heisenberg inequality. In Section \ref{fourier} together with Appendix \ref{CirCSqangle},  we discuss the link between  the coherent states on which our work is based and the coherent states built from probabilistic requirements given in  \cite{argaho12}. The latter are a generalisation of coherent states for the circle proposed by various authors \cite{main:ch5:debgo,main:ch5:kopap,main:ch5:delgo,main:ch5:kowrem1,main:ch5:hallmitch}, as well as of subsequent developments \cite{main:ch5:kowrem2,trifonov03,main:ch5:kowrem3}. In the conclusion (Section \ref{conclu}), we give some hints on upcoming research. 
%---------------------------------------------------------------------------------------
\section{The general setting}\label{sec:gralset}
This material is essentially borrowed from \cite{aagbook14,debievre89}.
\subsection{Semidirect product groups}
\label{sub:semi}
Let us consider an $n$-dimensional vector space $V$, a subgroup $S$ of $GL(V)$ and the group $G=V\rtimes S$ with:
\begin{itemize}
\item the action $v\mapsto sv$ of $S$ on $V$, for $v\in V$ and $s\in S$,

\item the semidirect product law of
composition $(x_{1},s_{1})(x_{2},s_{2})=(x_{1}+s_{1}x_{2},s_{1}s_{2})$ for
$x_{1},x_{2}\in V$ and $s_{1},s_{2}\in S$,

\item the action $V^{*}\ni k\mapsto sk$ of $S$ on the dual $V^{*}\sim V$,  defined
by  $\lg sk;x\rg =\lg k;s^{-1}x\rg$ (dual pairing between  $V^{*}$ and $V$),

\item the \textit{adjoint action} of $G$ on its Lie algebra $\mathfrak{g}$: $\text{Ad}_{g}(X)=gXg^{-1}$ for $g\in G$ and $X\in\mathfrak{g}$,

\item the \textit{coadjoint action} of $G$ on $\mathfrak{g}^{*}$: $\lg \text{Ad}_{g}^{\#}(X^{*});X\rg_{\mathfrak{g}^{*},\mathfrak{g}} =\lg X^{*};\text{Ad}_{g^{-1}}(X)\rg_{\mathfrak{g}^{*},\mathfrak{g}}$ for $X^{*}$ (dual pairing between  $\mathfrak{g}^{*}$ and $\mathfrak{g}$).
\end{itemize}
We now present some useful isomorphisms (summarized in the equation \ref{eq:isomorph}). Given the orbit of $k_{0}\in V^{*}$ under the action of $S$
\begin{equation}
\mathcal{O}^{*}=\lbrace k=sk_{0}\in V^{*}\,\vert\,s\in S\rbrace\, ,
\end{equation}
 the \textit{cotangent bundle} $T^{*}\mathcal{O}^{*}:=\bigcup_{k\in\mathcal{O}^{*}}T^{*}_{k}\mathcal{O}^{*}$ admits a symplectic structure. Given the Lie algebras $\mathfrak{v}$ of $V$ and  $\mathfrak{s}$ of $S$ respectively, the (coadjoint) orbit $\mathcal{O}^\ast_{(k_{0},0)}=\lbrace \text{Ad}_{g}^{\#}(k_{0},0)\in\mathfrak{g}^{*}\,\vert\, g\in G\rbrace$ of $(k_{0},0)\in\mathfrak{g}^{*}$ (for $k_{0}\in\mathfrak{v}^{*}$ and $0\in\mathfrak{s}^{*}$)  is isomorphic to $T^{*}\mathcal{O}^{*}$ under the coadjoint action \cite{aagbook14} . The stabilizer $H_{0}=N_{0}\rtimes S_{0}$ of $(k_{0},0)\in\mathfrak{g}^{*}$ under the coadjoint action is the semi-direct product between the \textit{annihilator}
\begin{equation}
N_{0}=\left\lbrace x\in V:\lg p;x\rg=0, \forall p\in T^{*}_{k_{0}}\mathcal{O}^{*}
\right\rbrace\, ,
\end{equation}
and the stabilizer $S_{0}=\left\lbrace s\in S \vert sk_{0}=k_{0}\right\rbrace $ of $k_{0}\in V^{*}$ under the action of $S$. The left coset space $X=G/H_{0}$ is isomorphic to $T^{*}\mathcal{O}^{*}$. Considering the space $V_{0}=T^{*}_{k_{0}}\mathcal{O}^{*}$, the space $T^{*}\mathcal{O}^{*}$ is isomorphic to $V_{0}\times\mathcal{O}^{*}$ as a Borel space.  We can summarize the isomorphisms given above in the following way \footnote{A detailed proof of this relation can be found in \cite{aagbook14}, Chapter 9, Section 9.2.2}
\begin{equation}
\label{eq:isomorph}
\mathcal{O}_{(k_{0},0)}\simeq T^{*}\mathcal{O}^{*}\simeq X=  G/H_{0} \simeq V_{0}\times\mathcal{O}^{*}\, . 
\end{equation}
\subsection{Induced representations for semi-direct product groups}
Let us consider a one-dimensional unitary representation of $V$ given by the character $\chi(v)=\exp(-\ii\langle k_{0};v\rangle)$ (for $k_{0}\in V^{*}$ and $v\in V$), and a unitary irreducible representation $s\mapsto L(s)$ of $S_{0}$ (carried by the Hilbert space $\mathcal{K}$). Then one defines a unitary irreducible representation of $V\rtimes S_{0}$ as
\begin{equation}
\label{eq:rep1}
\left(\chi \otimes L\right)(v,s)=e^{-\ii\langle k_{0};v\rangle}L(s)\quad\text{carried by}\,\,\mathcal{K}\,.
\end{equation}
Given the relation $G/(V\rtimes S_{0})\simeq\mathcal{O}^{*}$, one \textit{induces} a representation of $G$ from the representation $\chi \otimes L$ of $V\rtimes S_{0}$.  Let us consider the bundle $S\xrightarrow{\pi_{S}}\mathcal{O}^{*}$ with the projection $S\ni s\mapsto \pi_{S}(s)\in\mathcal{O}^{*}$, and the smooth section $\Lambda:\mathcal{O}^{*}\rightarrow S$ such that
\begin{subequations}
\label{eq:lambdasection}
\begin{align}
\Lambda(k_{0})&=e\;(\text{identity element of}\;S)\\
\Lambda(k)k_{0}&=k,\;k\in\mathcal{O}^{*}\, .
\end{align}
\end{subequations}
In this way any element $s\in S$ can be written as
\begin{equation}
s=\Lambda(k)s_{0}\quad\text{for}\,\,k\in\mathcal{O}^{*},\,\,s_{0}\in S_{0}\, .
\end{equation}
Then one defines the action of $S$ on $\mathcal{O}^{*}$ as $s^{\prime}\pi_{S}(s)=\pi_{S}(s^{\prime}s)$ for $s^{\prime},s\in S$. Considering this action and the property 
$\pi_{S}\left[\Lambda(sk)\right]=sk$ one gets
\begin{equation}
\label{eq:eqproj}
\pi_{S}\left[\Lambda(sk)\right]=\pi_{S}\left[s\Lambda(k)\right]\, .
\end{equation}
Let us consider the bundle $G\xrightarrow{\pi_{G}}\mathcal{O}^{*}$ with the projection $\pi_{G}(x,s)=\pi_{S}(s)$. A smooth section $\lambda:\mathcal{O}^{*}\rightarrow G$ is defined by
\begin{equation}
\label{eq:tinylambda}
\lambda(k)=(0,\Lambda(k))\, .
\end{equation}
According to the definition (\ref{eq:eqproj}) and the equation (\ref{eq:tinylambda}), one finds the relation $\pi_{G}\left[(0,\Lambda(sk))\right]=\pi_{G}\left[(v,s\Lambda(k))\right]$. In other words $(0,\Lambda(sk))$ and $(v,s\Lambda(k))$ belong to the same fiber (equivalence class). Therefore, 
\begin{equation}
(0,\Lambda(sk))h((v,s),k)=
(v,s\Lambda(k))\, .
\end{equation}
The element $h((v,s),k)\in G$ defines the cocycles\footnote{For $g,g_{1},g_{2}\in G$ and $k\in\mathcal{O}^{*}$ with $h^{\prime}(g,k)=[h(g^{-1},k)]^{-1}\in V\rtimes S_{0}$, the cocycle conditions are: \[\left\lbrace\begin{array}{rl}
h^{\prime}(g_{1}g_{2},k)&=h^{\prime}(g_{1},k)h^{\prime}(g_{2},g_{1}^{-1}k),\\
h^{\prime}(e,k)&=e.
\end{array}\right.\]} $h:G\times\mathcal{O}^{*}\rightarrow V\rtimes S_{0}$ and  $h_{0}:S\times\mathcal{O}^{*}\rightarrow S_{0}$ by
\begin{subequations}
\label{eq:cocycles}
\begin{align}
h((v,s),k)=&
(\Lambda(sk)^{-1}v,h_{0}(s,k))\, ,
\\
h_{0}(s,k)=&
\Lambda(sk)^{-1}s\Lambda(k).
\end{align}
\end{subequations}
Using the representation $\chi\otimes L$ of $V\rtimes S_{0}$, one represents $h((v,s)^{-1},k)\in V\rtimes S_{0}$ as
\begin{equation}
\label{eq:repind1}
(\chi \otimes L)(h((v,s)^{-1},k))=e^{-\ii\langle k_{0};v\rangle}L\left(h_{0}(s^{-1},k)\right)\, .
\end{equation}
Considering the space $\widetilde{\mathcal{H}}=\mathcal{K}\otimes L^2(\mathcal{O}^{*},d\nu)$  of all square-integrable functions $\phi:\mathcal{O}^{*}\rightarrow\mathcal{K}$  in the norm $\|\phi \|^{2}_{\widetilde{\mathcal{H}}}=\int_{\mathcal{O}^{*}}\|\phi(k) \|^{2}_{\mathcal{K}}d\nu(k)$, one defines the representation $(v,s)\mapsto \tensor[^\chi^L]{U}{}(v,s)$ of $G$ (carried by $\widetilde{\mathcal{H}}$) as
\begin{equation}
\label{eq:inducedrep}
\left( \tensor[^\chi^L]{U}{}(v,s)\phi\right)(k)=e^{\ii\langle k_{0};v\rangle}L\left(h_{0}(s^{-1},k)\right)^{-1}\phi(s^{-1}k)\, .
\end{equation}
The expression (\ref{eq:inducedrep}) is a representation of $G$, which  is \textit{induced} by the representation $\chi\otimes L$ of $V\rtimes S_{0}$. The representation $\tensor[^\chi^L]{U}{}$ of $G$ is irreducible.

\subsection{Coherent states for semi-direct product groups}
From the isomorphisms (\ref{eq:isomorph}) one constructs a section $V_{0}\times\mathcal{O}^{*}\ni(\bm{ p},\bm{ q})\mapsto\sigma(\bm{ p},\bm{ q})\in G$, where $(\bm{ p},\bm{ q})$ are canonically conjugate pairs for the symplectic structure of the manifold $V_{0}\times\mathcal{O}^{*}$. Given the invariant symplectic  measure $\ud\mu(\bm{ p},\bm{ q})$ for $V_{0}\times\mathcal{O}^{*}$, the action of the induced representation (\ref{eq:inducedrep}) of $G$ on a vector $\eta\in\widetilde{\mathcal{H}}$ gives a family of vectors $\eta^{\sigma}_{\bm{ p},\bm{ q}}\equiv \eta_{\bm{ p},\bm{ q}}\in\widetilde{\mathcal{H}}$ parametrized by $(\bm{ p},\bm{ q})$ (where $\bm{q}\in V_{0}$ and $\bm{p}\in\mathcal{O}^{*}$)
\begin{equation}
\label{eq:generalcoherent}
\eta_{\bm{ p},\bm{ q}}(k)=\left( \tensor[^\chi^L]{U}{}(\sigma(\bm{ p},\bm{ q}))\eta\right)(k)
\quad
\Leftrightarrow
\quad
\ket{\eta_{\bm{ p},\bm{ q}}}=\tensor[^\chi^L]{U}{}(\sigma(\bm{ p},\bm{ q}))\ket{\eta}\, .
\end{equation}
Since in the present paper the representation $L$ is actually trivial,  we  dismiss $\mathcal{K}$ from now on, so that $\widetilde{\mathcal{H}} = \mathcal{H}=L^2(\mathcal{O}^{*},d\nu)$. Let us consider the formal integral
\begin{equation}
\label{eq:formalintegral}
\int_{V_{0}\times\mathcal{O}^{*}}\ud\mu(\bm{ p},\bm{ q})\scalar{\phi}{\eta_{\bm{ p},\bm{ q}}}_{\mathcal{H}}\scalar{\eta_{\bm{ p},\bm{ q}}}{\psi}_{\mathcal{H}}\,, \,  \mathrm{where} \,\,\phi,\psi:\mathcal{O^*}\rightarrow \C \,.
\end{equation}
If we prove that it is equal to $c_{\eta}\scalar{\phi}{\psi}$ for some constant $0<c_{\eta}<\infty$, we  obtain that   the resolution of the identity
\begin{equation}
\label{eq:resolutionof}
\dfrac{1}{c_{\eta}}\int_{V_{0}\times\mathcal{O}^{*}}\ud\mu(\bm{ p},\bm{ q})\ket{\eta_{\bm{ p},\bm{ q}}}\bra{\eta_{\bm{ p},\bm{ q}}}=I\quad\text{where}
\,\,0<c_{\eta}<\infty
\end{equation}
holds on $\mathcal{H}$. In the case we are considering in this paper, we will see that  \eqref{eq:resolutionof} holds by imposing  restrictions on $\text{supp}\,\eta$. When (\ref{eq:resolutionof}) is valid, the states (\ref{eq:generalcoherent}) are our  (covariant) \textit{coherent states}, which generalize  the Gilmore-Perelomov construction \cite{gilmore72,perelomov86}.
%--------------------------
\subsection{Covariant integral quantisation}

In the present framework, CS quantisation maps the classical observable $f(\bm{ p},\bm{ q})$ in  phase space to the operator $A_{f}$ acting on the Hilbert space $\mathcal{H}$  by means of the formula
\begin{equation}
\label{eq:integralquant}
A^{\sigma}_{f} \equiv A_{f}=
\dfrac{1}{c_{\eta}}
\int_{V_{0}\times\mathcal{O}^{*}}\ud\mu(\bm{ p},\bm{ q})
\left\vert\eta_{\bm{ p},\bm{ q}}\right\rangle
\left\langle\eta_{\bm{ p},\bm{ q}}\right\vert f(\bm{ p},\bm{ q})\, .
\end{equation}
If $f$ is real, the operator $A_{f}$ is symmetric by construction, and if $f$ is real semi-bounded, then there exists a canonical  self-adjoint extension of $A_f$, called Friedrichs extension,   based on quadratic forms. 

%Moreover, a new family of states $\eta_{g\sigma(\bm{ p},\bm{ q})}\in\mathcal{H}$ is defined by the action of $g\in G$ on  $\eta_{\bm{ p},\bm{ q}}\in\mathcal{H}$ as
%\begin{equation}
%\label{eq:newcoherent}
%\ket{\eta_{g\sigma(\bm{ p},\bm{ q})}}=
%\tensor[^{\chi}]{U}{}(g)\ket{\eta_{\bm{ p},\bm{ q}}}\, .
%\end{equation}
%Since $I=\tensor[^{\chi}]{U}{}(g)\tensor[^{\chi}]{U}{}(g)^{*}$, it is easy to see that
%\begin{equation}
%\label{eq:newresolution}
%\dfrac{1}{c_{\eta}}
%\int_{V_{0}\times\mathcal{O}^{*}}\ud\mu(p,q)
%\ket{\eta_{g\sigma(\bm{ p},\bm{ q})}}
%\bra{\eta_{g\sigma(\bm{ p},\bm{ q})}}
%=I\, .
%\end{equation}
%Thus, because the family of states (\ref{eq:newcoherent}) satisfies the resolution of the identity on $\widetilde{\mathcal{H}}$, they are also \textit{coherent states}. Integral quantisation allows us to define a new operator $A_{f}^{g\sigma(\bm{ p},\bm{ q})}$  using the coherent states from (\ref{eq:newcoherent}),
%\begin{equation}
%\label{eq:newoperator}
%A_{f}^{g\sigma}=
%\dfrac{1}{c_{\eta}}
%\int_{V_{0}\times\mathcal{O}^{*}}\ud\mu(\bm{ p},\bm{ q})
%\ket{\eta_{g\sigma(\bm{ p},\bm{ q})}}
%\bra{\eta_{g\sigma(\bm{ p},\bm{ q})}}
%f(\bm{ p},\bm{ q})\, .
%\end{equation}
%Covariance of this quantisation procedure can be readily established from  (\ref{eq:newcoherent}) and (\ref{eq:newoperator}),
%\begin{equation}
%\label{eq:covariantintegarl}
%\tensor[^{\chi}]{U}{}(g)A_{f}\tensor[^{\chi}]{U}{}(g)^{*}
%=
%A_{f}^{g\sigma(\bm{ p},\bm{ q})}\, .
%\end{equation}
Covariance holds in the sense of \eqref{quantizcovmodH2}:
\begin{equation}
\label{quantizcovqp}
 \tensor[^{\chi L}]{U}{}(g) A_f \tensor[^{\chi L}]{U}{}(g)^{\dagger} = A^{\sigma_g}_{\mathcal{U}_l(g)f}\, , \quad A^{\sigma_g}_f:= \dfrac{1}{c_{\eta}}
\int_{V_{0}\times\mathcal{O}^{*}}\ud\mu_{\sigma_g}(\bm{ p},\bm{ q})
\left\vert\eta^{\sigma_g}_{\bm{ p},\bm{ q}}\right\rangle
\left\langle\eta^{\sigma_g}_{\bm{ p},\bm{ q}}\right\vert f(\bm{ p},\bm{ q})\,, 
\end{equation} 
with  $\ket{\eta^{\sigma_g}_{\bm{ p},\bm{ q}}}=\tensor[^\chi^L]{U}{}(g\sigma(g^{-1}(\bm{ p},\bm{ q})))\ket{\eta}$ and  $\mathcal{U}_l(g)f(\bm{ p},\bm{ q})= f\left(g^{-1}(\bm{ p},\bm{ q})\right)$.

In the context of  CS quantisation,  one defines the semiclassical portrait of the operator $A_{f}$, or its lower \cite{lieb73} or covariant \cite{berezin74} symbol, as 
\begin{equation}
	\label{lower1111}
	\check{f}(\bm{ p},\bm{ q})=Tr\left(\ket{\eta_{\bm{ p},\bm{ q}}}\bra{\eta_{\bm{ p},\bm{ q}}}A_{f}\right)=
	\dfrac{1}{c_{\eta}}
	\int_{V_{0}\times\mathcal{O}^{*}}\ud\mu(\bm{ p}^{\prime},\bm{ q}^{\prime})f\left(\bm{ p}^{\prime},\bm{ q}^{\prime}\right)
	\left|\left\langle \eta_{\bm{ p}^{\prime},\bm{ q}^{\prime}}|\eta_{\bm{ p},\bm{ q}}\right\rangle \right|^{2}\, .
\end{equation}
It can be viewed as  the average of the function $f(\bm{ p},\bm{ q})$ with respect to the probability distribution $(\bm{ p}^{\prime},\bm{ q}^{\prime}) \mapsto \left|\left\langle \eta_{\bm{ p}^{\prime},\bm{ q}^{\prime}}|\eta_{\bm{ p},\bm{ q}}\right\rangle \right|^{2}/c_{\eta}$ with respect to the measure $\mu$.
Given a family of scale parameters $\epsilon_i$, e.g. Planck constant, characteristic length, ...,  and a distance function $d(f,\check{f})$, a classical
limit of $A_f$, if it exists, can be viewed as
\begin{equation}
d(f,\check{f})\rightarrow 0\quad \text{as}\quad \epsilon_i\rightarrow 0\, ,
\end{equation}
while the set of $\epsilon_i$'s are subject to possible constraints, like fixed ratios. 
%---------------------------------------------------------------------------------------
\section{Coherent states for $\text{E}(2)$}\label{sec:cohe}
Following the method exposed in the previous section, one now considers the Euclidean group $G=\text{E}(2)$, where $V=\R^{2}$ and $S=\text{SO}(2)$. The action of $\text{SO}(2)$ on $\R^{2}$ is
\begin{equation}\label{eq:actionofso2}
\bm{x}\mapsto \mathcal{R}(q)\bm{x},\quad \bm{x}\in\R^{2},\,\mathcal{R}(q)\in\text{SO}(2),\,q\in[0,2\pi)\,,
\end{equation}
where $\mathcal{R}(q)$ corresponds to a rotation by the angle $q$. An element of $\text{E}(2)$ can be represented by the pair $(\bm{x},q)$, for which the composition law  reads
\begin{equation}
\label{x0000}
(\bm{x},q)(\bm{x}^{\prime},q^{\prime})=(\bm{x}+\mathcal{R}(q)\bm{x}^{\prime},
q+q^{\prime})\, .
\end{equation}
Since $V^{*}=\R^{2}$, for $\bm{k}\in V^{*}$ and $\bm{x}\in V$ the dual pairing is $\langle\bm{k};\bm{x}\rangle=\bm{k}\cdot\bm{x}=k_{1}x_{1}+k_{2}x_{2}$. The orbit of $\bm{k}_{0}\in\R^{2}$ under the action of $\text{SO}(2)$ is $\mathcal{O}^{*}=\lbrace\bm{k}=\mathcal{R}(q)\bm{k}_{0}\in\R^{2}\, |\,\mathcal{R}(q)\in\text{SO}(2)\rbrace$, i.e., it is a circle. In particular, if $\lVert\bm{k}_{0}\rVert=1$, we have $\mathcal{O}^{*}=\mathbb{S}^{1}$. The induced representation $\tensor[^\chi^L]{U}{} \equiv U$ of $\text{E}(2)$ is obtained by direct application of Equation (\ref{eq:inducedrep}). Since $h_{0}:\,S\times\mathcal{O}^{*}\rightarrow S_{0}$ and the stabilizer $S_{0}$ of $k_0$ is the identity, for $(\bm{r},\theta)\in\text{E}(2)$ we have
\begin{equation}
\label{eq:repforso2}
\left( U(\bm{r},\theta)\phi\right)(\alpha)=e^{\ii(r_{1}\cos\theta+r_{2}\sin\theta)}\phi(\alpha-\theta)\, ,
\end{equation}
where $\phi(\alpha)\in L^{2}(\mathbb{S}^{1},\ud\alpha)$. \\
The stabilizer under the \textit{coadjoint action} $\text{Ad}_{\text{E}(2)}^{\#}$ is
\begin{equation}
\label{eq:stabilazerH0}
H_{0}=
\left\lbrace (\bm{x},0)\in\text{E}(2)\,\vert\,\bm{c}\cdot\bm{x}=0,\,\bm{c}\in\R^{2}\,\text{fixed} \right\rbrace\, , 
\end{equation}
and $V_0 \sim \R$.

The \textit{cotangent bundle} represents the classical phase space for a particle moving on a circle, where $q$ is the angular position. Due to the isomorphisms (\ref{eq:isomorph}), the cotangent bundle becomes
\begin{equation}
\label{eq:newisomorph}
X \equiv T^{*}\mathbb{S}^{1}\simeq(\R^{2}\rtimes\text{SO}(2))/H_{0}\simeq\R\times\mathbb{S}^{1}\, ,
\end{equation}
that is, it carries coordinates $(p,q)\in \mathbb{R}\times \mathbb{S}^{1}$ and  symplectic invariant measure $\ud p\, \ud q\equiv \mathrm{d}p\wedge \mathrm{d}q$.
\begin{theo}
There exists a section $\sigma:\R\times\mathbb{S}^{1}\rightarrow\text{E}(2)$ defined as
\begin{equation}
\label{x0021}
\sigma(p,q)=(\mathcal{R}( q)(\bm{\kappa}p+\bm{\lambda}),q)\, ,
\end{equation}
where $\bm{\kappa},\bm{\lambda}\in\R^{2}$ are constant vectors.
\end{theo}
\begin{proof}
Let us consider a section $\sigma:\R\times\mathbb{S}^{1}\rightarrow\text{E}(2)$ given by
\begin{equation}
\label{eq:gralsection}
\sigma(p,q)=(\bm{f}(p,q),q)\, ,
\end{equation}
where $\bm{f}(p,q)$ is a function to be determined. An immediate factorization of $(\bm{r},\theta)\in G$
using the elements $(\bm{x},0)\in H_{0}$ is given by
\begin{equation}
\label{x0002}
(\bm{r},\theta)=(\bm{r}^{\prime},\theta)(\bm{x},0),\quad\text{for some }\bm{r}^{\prime}\in \R^{2}\, .
\end{equation}
With this expression at hand,  one writes \footnote{For $g\in G$, $\gamma\in X $ and $h\in H_{0}$, the isomorphisms (\ref{eq:isomorph}) validate the relation $g\sigma(\gamma)=\sigma(g\gamma)h$, where $g\gamma$ is the action of $G$ on $X$. We make the following correspondences: $\sigma(p,q)=\sigma(\gamma)$ and $\sigma(p^{\prime},q^{\prime})=\sigma(g\gamma)$}
\begin{equation}
\label{x0007}
(\bm{r},\theta)\sigma(p,q)=\sigma(p^{\prime},q^{\prime})(\bm{x},0)\, . 
\end{equation}
Taking into account  section (\ref{eq:gralsection}), the decomposition (\ref{x0007}) becomes
\begin{equation}
\label{x0009}
(\bm{r}+\mathcal{R}( \theta)\bm{f}(p,q),\theta+q)=
(\bm{f}(p^{\prime},q^{\prime})+\mathcal{R}( q^{\prime})\bm{x},q^{\prime})\, . 
\end{equation}
Therefore, we arrive at the conditions
\begin{subequations}
\label{x0010}
\begin{align}
\bm{r}+\mathcal{R}( \theta)\bm{f}(p,q)&=\bm{f}(p^{\prime},q^{\prime})+\mathcal{R}( q^{\prime})\bm{x}\, ,\\
q^{\prime}&=\theta+q\, .
\end{align}
\end{subequations}
These conditions determine  the change of variables $(p,q)\rightarrow(p^{\prime},q^{\prime})$  for the angular coordinate $q$ and its conjugate momentum $p$. The condition (\ref{x0010}) provides a definition for the function $q^{\prime}(q,\theta)$, but not an explicit expression for the function $p^{\prime}(p,q,\theta,\bm{r})$. From (\ref{x0010}) the vector $\bm{x}$ is written as
\begin{equation}
\label{x0011}
\bm{x}=
\mathcal{R}( -q-\theta)\bm{r}+\mathcal{R}( -q)\bm{f}(p,q)-
\mathcal{R}( -q-\theta)\bm{f}(p^{\prime}(p,q,\theta,\bm{r}),q^{\prime}(q,\theta))\, .
\end{equation}
Since $(\bm{x},0)\in H_{0}$, we have from \eqref{eq:stabilazerH0} $\bm{c}\cdot\bm{x}=0$. For the particular case $\bm{r}=\bm{0}$, equation \eqref{x0011} is written as
\begin{equation}
\label{x0012}
\bm{c}\cdot\left[
\mathcal{R}( -q)\bm{f}(p,q)-
\mathcal{R}( -q-\theta)\bm{f}\left(p^{\prime}(p,q,\theta,\bm{0}) ,q+\theta\right)
\right]
= 0 \, , \quad \mbox{for all} \quad q\,, p\, , \theta\, . 
\end{equation}
Choosing $q=0$ leads to 
\begin{equation}
\label{xgg12}
\bm{c}\cdot\left[
\bm{f}(p,0)
\mathcal{R}( -\theta)\bm{f}\left(p^{\prime}(p,0,\theta,\bm{0}) ,\theta\right)
\right]
= 0\, .
\end{equation}
The vector $\bm{c}$ is fixed, and when the value of $p$ is also fixed, the right-hand side of the equation (\ref{xgg12}) should be independent of $\theta$. Therefore $p^{\prime}(p,0,\theta,\bm{0})$ must be independent of $\theta$, i.e.,  $p^{\prime}=p^{\prime}(p,0,\bm{0})$. In order to eliminate the dependence on $\theta$ from (\ref{xgg12}), we write
$\bm{f}\left(p^{\prime}(p,0,\bm{0}),\theta\right)$ as
\begin{equation}
\label{x0013}
\bm{f}\left(p^{\prime}(p,0,\bm{0}),\theta\right)
=
\mathcal{R}( \theta)\bm{g}(p^{\prime}(p,0,\bm{0}))\, ,
\end{equation}
where $\bm{g}$ is a function to be determined.
Since  $\ud p\wedge \ud q$ should be left invariant under the change  of variables $(p,q) \mapsto (p^{\prime},q^{\prime}) $, i.e. $\ud\,p^{\prime}\,\ud q^{\prime}=\vert J\, \vert \ud p\,\ud q$ with  $\vert J\vert =1$, the only possible choice is $\frac{\partial}{\partial p}p^{\prime}(p,0,\bm{0})=1$.  Therefore $p^{\prime}(p,0,\bm{0})$
must be
\begin{equation}
\label{x0016}
p^{\prime}(p,0,\bm{0})=p + \text{constant}\, .
\end{equation}
Considering (\ref{x0013}) and (\ref{x0016}) we arrive at
\begin{equation}
\label{xgg16}
\bm{f}(p + \text{constant},\theta)=\mathcal{R}( \theta)\bm{g}(p + \text{constant})\, .
\end{equation}
 The simplest generalisation of \eqref{xgg16} to the case $q\neq 0$ is
\begin{equation}
\label{xgg16i}
\bm{f}(p + \text{constant},q+\theta)=\mathcal{R}( q+\theta)\bm{g}(p + \text{constant})\, .
\end{equation}
With this choice,  the vector $\bm{f}(p,q)$ assumes the form
\begin{equation}
\label{xgg16ii}
\bm{f}(p,q)=\mathcal{R}( q)\bm{g}(p)\, .
\end{equation}
Now $\bm{g}(p)$ has the property
$\bm{g}(p)\rightarrow\bm{g}(p+ \text{constant})$ when $p\rightarrow p^{\prime}$. Hence,  the simplest choice is $\bm{g}(p)=\bm{\kappa}p+\bm{\lambda}$ where $\bm{\kappa},\bm{\lambda}\in\R^{2}$ are constant vectors
\begin{equation}
\label{x0020}
\bm{f}(p,q)=\mathcal{R}( q)(\bm{\kappa}p+\bm{\lambda})\, .
\end{equation}
The section $\sigma(p,q)$ given by (\ref{x0021}) may not be the most general type of Borel section allowed in this problem, but is compatible with the conditions (\ref{x0010}).
\end{proof}
\begin{defi}
\label{defiL2S}
We denote by $L^{2}(\mathbb{S}^{1},\ud\alpha)$ the Hilbert space of $2\pi$-periodic complex-valued functions $\psi(\alpha)$ which are square-integrable on a period interval $\alpha_0, \alpha_0 + 2\pi$, $\alpha_0 \in \R$, 
\begin{equation}
\label{sqint}
\int_{\alpha_0}^{\alpha_0 + 2\pi}\ud\alpha \,\vert \psi(\alpha)\vert^2\equiv \int_{\mathbb{S}^{1}}\ud\alpha \,\vert \psi(\alpha)\vert^2\, ,  
\end{equation}
and equipped with the scalar product
\begin{equation}
\label{l2scal}
\lg \phi|\psi\rg= \int_{\mathbb{S}^{1}}\ud\alpha \,\overline{\psi(\alpha)}\,\phi(\alpha)\, . 
\end{equation}
\end{defi}
\begin{defi}
\label{defics}
Using the section (\ref{x0021}), the induced representation (\ref{eq:repforso2}), and a choice of fiducial vector $\eta\in L^2(\mathbb{S}^1,\ud \alpha)$, we define the following  states in the manner of (\ref{eq:generalcoherent}):
\begin{equation}
\label{qpalfa}
\eta_{p,q}  (\alpha)=
 e^{\ii[\mathcal{R}(q -\alpha)(\bm{\kappa}p+\bm{\lambda})]_{1}}\eta(\alpha-q)= e^{\ii[\kappa p\cos(q-\alpha + \gamma) +\lambda \cos(q-\alpha + \zeta)]}\eta(\alpha-q)\, ,
\end{equation} 
with $\bm{\kappa}= \kappa \begin{pmatrix}
      \cos\gamma   \\
      \sin\gamma
\end{pmatrix}$  and $\bm{\lambda}= \lambda \begin{pmatrix}
      \cos\zeta   \\
      \sin\zeta
\end{pmatrix}$,  $\kappa= \Vert \bm{\kappa}\Vert$,  $\lambda= \Vert \bm{\lambda}\Vert$, $\gamma=\arg \bm\kappa$, $\zeta= \arg \bm\lambda$.
%where we have used $[\mathcal{R}( -\alpha)\bm{r}]_{1}=r_{1}\cos\alpha+r_{2}\sin\alpha$.
\end{defi}
\begin{defi}
\label{defi1}
With the same notations as above 
we define the function $S_{x}(\alpha, \alpha^{\prime}, q)$ as
\begin{equation}
\label{Saa}
S_{x}(\alpha, \alpha^{\prime},q)=
\sin\left( \dfrac{\alpha-\alpha^{\prime}}{2}\right)\,
\sin\left( q + x - \dfrac{\alpha+\alpha^{\prime}}{2}\right)\,,\quad \text{for}\quad x=\gamma,\zeta\, .
\end{equation}
\end{defi}
We now have to prove that the states \eqref{qpalfa} are coherent in the sense that they solve the identity. 
\begin{theo}
\label{theorem1}
Given two functions $\psi,\phi\in L^{2}(\mathbb{S}^{1},\ud\alpha)$, their scalar product $\scalar{\phi}{\psi}$ is equal to the  integral
\begin{equation}
\begin{aligned}
\label{x0030}
I(\psi,\phi)=\int_{\mathbb{R}\times \mathbb{S}^{1}}\dfrac{\mathrm{d}p\,\mathrm{d}q}{c_{\eta}}\int_{\mathbb{S}^{1}}\ud\alpha\,
\overline{\psi(\alpha)}\,\eta_{p,q}  (\alpha)
\int_{\mathbb{S}^{1}}\ud\alpha^{\prime}\,
\overline{\eta_{p,q}  (\alpha^{\prime})}\,\phi(\alpha^{\prime})\, ,
\end{aligned}
\end{equation}
and the vectors $\eta_{p,q}$ form a family of \textit{coherent states} for $\text{E}(2)$ which resolves the identity on $L^{2}(\mathbb{S}^{1},\ud\alpha)$,
\begin{equation}
\label{resuneta}
I=\int_{\mathbb{R}\times \mathbb{S}^{1}}\dfrac{\mathrm{d}p\,\mathrm{d}q}{c_{\eta}}\,\ket{\eta_{p,q}}\bra{\eta_{p,q}}\, ,
\end{equation}
 if $\eta(\alpha)$ is \underline{admissible} in the sense that $\textbf{supp}\,\eta \in (\gamma-\pi,\gamma)\mathrm{mod}\,2\pi$ , and 
\begin{equation}
\label{eqn:idntresol}
0< c_{\eta}:=\dfrac{2\pi}{\vert\bm\kappa\vert}
\int_{\mathbb{S}^{1}}
\dfrac{|\eta(q)|^{2}}{ \sin(\gamma - q)}\,\mathrm{d} q <\infty\, .
\end{equation}
\end{theo}
\begin{proof}
 In order to prove theorem \ref{theorem1}, we must find the conditions for which the integral $I(\psi,\phi)$ is finite and equal to $\scalar{\psi}{\phi}$. After integrating with respect to the variable $p$ by using $\int_{\R}\ud p e^{-\ii p k}= 2\pi \delta(k)$, the integral (\ref{x0030}) becomes
\begin{equation}
\label{integral2}
\begin{aligned}
I(\psi,\phi)=\dfrac{2\pi}{c_{\eta}}\int_{\mathbb{S}^1 }\ud q\,\int_{\mathbb{S}^1 }\ud\alpha^{\prime}\,\overline{\eta(\alpha^{\prime}-q)}\phi(\alpha^{\prime})\int_{\mathbb{S}^1 }\ud\alpha\,\overline{\psi(\alpha)}\eta(\alpha-q)e^{2\ii\lambda S_{\zeta}(\alpha,\alpha^{\prime},q)}
\\ \times
\delta(2\kappa S_{\gamma}(\alpha,\alpha^{\prime},q))\, .
\end{aligned}
\end{equation}
Now the Dirac delta has the expansion
\begin{equation}
\begin{aligned}
\delta(2\kappa S_{\gamma}(\alpha,\alpha^{\prime},q))=\sum_{k}\dfrac{\delta\left(\alpha-\alpha_{k}\right)}{2\kappa\left|\partial_{\alpha}S_{\gamma}(\alpha,\alpha^{\prime},q)\right|_{\alpha=\alpha_{k}}}+\sum_{k^{\prime}}\dfrac{\delta\left(\alpha-\alpha_{k^{\prime}}\right)}{2\kappa\left|\partial_{\alpha}S_{\gamma}(\alpha,\alpha^{\prime},q)\right|_{\alpha=\alpha_{k^{\prime}}}}\, ,
\end{aligned}
\end{equation}
where $\alpha_{k}$ and $\alpha_{k^{\prime}}$ are the roots of $S_{\gamma}(\alpha,\alpha^{\prime},q)$ obtained when $\dfrac{\alpha-\alpha^{\prime}}{2}=k\pi$ or $q+\gamma-\dfrac{\alpha+\alpha^{\prime}}{2}=k^{\prime}\pi$ for $k,k^{\prime}\in\mathbb{Z}$. Hence, $\alpha_{k}=\alpha^{\prime}+2k\pi$ and $\alpha_{k^{\prime}}=-\alpha^{\prime}+2q+2\gamma-2k^{\prime}\pi$. The Dirac delta is now written as
\begin{equation}
\begin{aligned}
\label{delta1}
\delta(2\kappa S_{\gamma}(\alpha,\alpha^{\prime},q))=\sum_{k}\dfrac{\delta\left(\alpha-\alpha_{k}\right)}{\kappa\left|\sin\left(q+\gamma-\alpha^{\prime}\right)\right|}+\sum_{k^{\prime}}\dfrac{\delta\left(\alpha-\alpha_{k^{\prime}}\right)}{\kappa\left|\sin\left(q+\gamma-\alpha^{\prime}\right)\right|}\, .
\end{aligned}
\end{equation}
With the help of expression (\ref{delta1}), using the $2\pi$ periodicity of all involved functions and the fact that one integrates over one period interval,  integral (\ref{integral2}) becomes
\begin{equation}
\begin{aligned}
I(\psi,\phi)=\dfrac{2\pi}{\kappa c_{\eta}}\int_{\mathbb{S}^1 }\ud\alpha^{\prime}\,\overline{\psi(\alpha^{\prime})}\phi(\alpha^{\prime})\int_{\mathbb{S}^1 }\ud q\,\dfrac{\overline{\eta(\alpha^{\prime}-q)}\eta(\alpha^{\prime}-q)}{\left|\sin\left(\gamma-(\alpha^{\prime}-q)\right)\right|}
\\
+\dfrac{2\pi}{\kappa c_{\eta}}\int_{\mathbb{S}^1 }\ud q\,\int_{\mathbb{S}^1 }\ud\alpha^{\prime}\,\dfrac{\overline{\eta(\alpha^{\prime}-q)}\phi(\alpha^{\prime})}{\kappa\left|\sin\left(q+\gamma-\alpha^{\prime}\right)\right|}e^{2\ii\lambda \sin\left(\gamma + q-\alpha^{\prime}\right)
\sin\left(\zeta -\gamma \right)}
\\
\times\overline{\psi(-\alpha^{\prime}+2q+2\gamma)}
\eta\left(q-\alpha^{\prime}+2\gamma\right)\, .
\end{aligned}
\end{equation}
Performing the change of variable $q\mapsto q'=\alpha^{\prime}-q$ in both integrals, and choosing $(\gamma -\pi, \gamma + \pi)$ as the integration interval for the $q^{\prime}$ variable, one has
\begin{equation}
\label{suppeqn}
\begin{aligned}
I(\psi,\phi)=\dfrac{2\pi}{\kappa c_{\eta}}\int_{\mathbb{S}^1 }\ud\alpha^{\prime}\,\overline{\psi(\alpha^{\prime})}\phi(\alpha^{\prime})\int_{\gamma-\pi}^{\gamma +\pi}\ud q^{\prime}\,\dfrac{\vert\eta(q^{\prime})\vert^2
}{\left|\sin\left(\gamma-q^{\prime}\right)\right|}
\\
+\dfrac{2\pi}{\kappa c_{\eta}}\int_{\mathbb{S}^1 }\ud \alpha^{\prime}\,\phi(\alpha^{\prime})\,\int_{\gamma-\pi}^{\gamma +\pi}\ud q^{\prime}\,\dfrac{\overline{\eta(q^{\prime})}}{\kappa\left|\sin\left(\gamma-q^{\prime}\right)\right|}e^{2\ii\lambda \sin\left(\gamma - q^{\prime}\right)
\sin\left(  \zeta -\gamma \right)}
\\
\times\overline{\psi(\alpha^{\prime}-2q^{\prime}+2\gamma)}
 \eta\left(2\gamma -q^{\prime}\right)\, .
\end{aligned}
\end{equation}
In order to avoid the singularity appearing in the denominator of the integrand of the first integral in (\ref{suppeqn}), we impose that $\left|\sin\left(\gamma-q^{\prime}\right)\right|\neq0$ for $q^\prime \in \text{supp}\,\eta$. Hence, we choose $\text{supp}\,\eta \subset (\gamma-\pi,\gamma)$. 
%\\
%\hspace*{2.5cm}
%\begin{tikzpicture}
%\draw (0,0)--(10,0);
%\draw [fill] (1.66,0) circle [radius=0.06];
%\draw (3.33,-.1)--(3.33,.1);
%\draw [fill] (5,0) circle [radius=0.06];
%\draw (6.66,-.1)--(6.66,.1);
%\draw [fill] (8.33,0) circle [radius=0.06];
%
%\node [below] at (1.66,0) {$\alpha^{\prime}-2\pi$};
%\node [below] at (3.33,0) {$0$};
%\node [below] at (5,0) {$\alpha^{\prime}$};
%\node [below] at (6.66,0) {$2\pi$};
%\node [below] at (8.33,0) {$\alpha^{\prime}+2\pi$};
%\end{tikzpicture}
%\vspace*{0.5cm}\\
%\noindent Therefore only $k=0$ is compatible with the condition on $\text{supp}\,\eta$. 
The second  integral vanishes, since $2\gamma-q^\prime \notin \text{supp}\,\eta$. Thus (\ref{suppeqn}) reduces to
\begin{equation}
\begin{aligned}
I(\psi,\phi)=\left\langle \psi|\phi\right\rangle \dfrac{1}{c_{\eta}}\dfrac{2\pi}{\kappa}
\int_{\gamma -\pi}^{\gamma}\ud q\,\dfrac{\vert \eta(q)\vert^2}{\sin\left(\gamma-q\right)}\, .
\end{aligned}
\end{equation}
Imposing the condition
\begin{equation}
\label{ceta}
c_{\eta}=\dfrac{2\pi}{\kappa}\int_{\gamma -\pi }^{\gamma}\ud q\,\dfrac{\vert \eta(q)\vert^2}{\sin\left(\gamma-q\right)}
<\infty\, ,
\end{equation}
gives $I(\psi,\phi)=\left\langle \psi|\phi\right\rangle$. With this result the integral (\ref{x0030}) take the form
\begin{equation}
\begin{aligned}
\label{x0theop}
\scalar{\psi}{\phi}=\int_{\mathbb{R}\times \mathbb{S}^{1}}\dfrac{\mathrm{d}p\mathrm{d}q}{c_{\eta}}
\scalar{\psi}{\eta_{p,q}}
\scalar{\eta_{p,q}}{\phi}\, .
\end{aligned}
\end{equation}
Hence the vectors $\eta_{p,q}$ form a family of \textit{coherent states} for $\text{E}(2)$ which resolves the identity on $L^{2}(\mathbb{S}^{1},\ud\alpha)$.
\end{proof}
%---------------------------------------------------------------------------------------

For the sake of later convenience, we introduce the following families of integrals: 
\begin{defi}
\label{deficeta}
Given a $2\pi$-periodic function $\eta(\alpha) \in L^2\left(\mathbb{S}^1,\ud \alpha\right)$ with $\text{supp}\,\eta \in (\gamma-\pi,\gamma)$, $\gamma \in [0,2\pi)$, we define the integrals,
\begin{equation}
\label{cetanu}
c_{\nu}(\eta,\gamma) = \int_{\mathbb{S}^1}\ud \alpha \frac{\vert \eta(\alpha)\vert^2}{(\sin(\gamma-\alpha))^{\nu}}\, ,
\end{equation}
where $\nu \in \C$ is such that convergence is assured.
\end{defi}
With this definition, $c_0= 1$ (normalisation of $\eta$), and the constant $c_{\eta}$ is given by $c_{\eta}= \frac{2\pi}{\kappa}\,c_{1}(\eta,\gamma)$.

%Finally, and from now on, we suppose that the admissible fiducial vector $\eta(\alpha)$ is bounded on its support. 
\begin{defi}
\label{defifjm}
For $\eta$ of class $C^k$,  we define for $j\leq k$ the set of functions 
\begin{equation}
\label{functionsfjm}
f_{j;m}(q)=\dfrac{\eta(q)\partial^{j}_{q}\overline{\eta(q)}}{(\sin(\gamma-q))^{m}}\, ,\quad \text{for}\ j,m\in\N\, .
\end{equation}
\end{defi}
%-----------------------------------------------------------------------------------------------------------------------------------------
\section{Quantisation of classical observables}\label{sec:quant}
The quantisation  of a classical observable $f(p,q)$
issued from (\ref{eq:integralquant}) with an admissible fiducial vector $\eta$ is given by
\begin{equation}
\label{x0039}
f\mapsto A_{f}=\int_{\R\times\mathbb{S}^{1}} \dfrac{\ud p\, \ud q}{c_{\eta}}
f(p,q)\,\ket{\eta_{p,q}}\bra{\eta_{p,q}}\, .
\end{equation}
The operator $A_{f}$  acts on the Hilbert
space $L^2( \S, \ud \alpha)$ as the integral operator 
\begin{equation}
\label{kerAf}
(A_{f}\psi)(\alpha) = \int_{\S}\ud \alpha^{\prime}\,\mathcal{A}_f(\alpha,\alpha^{\prime})
\, \psi(\alpha^{\prime})\, ,
\end{equation}
whose kernel $\mathcal{A}_{f}$ is  given by
\begin{equation}
\label{kernelAf}
\mathcal{A}_{f}(\alpha,\alpha^{\prime})=\dfrac{1}{c_{\eta}}\int_{\S}\ud q\,
\eta(\alpha-q)\overline{\eta(\alpha^{\prime}-q)}e^{2\ii\lambda S_{\zeta}(\alpha,\alpha^{\prime},q)}
\int_{-\infty}^{+\infty}\ud p\, e^{\ii 2\kappa S_{\gamma}(\alpha,\alpha^{\prime},q) p}\, f(p,q)\, .
\end{equation}
The expression \eqref{kernelAf} is quite involved. Hence, in the sequel we
examine manageable particular cases.

\subsection{Quantisation of a function of $q$}

Let us introduce the positive $2\pi$-periodic function 
\begin{equation}
\label{Eal}
E_{\eta;\gamma}(\alpha):= \frac{2\pi}{\kappa c_{\eta}} \,
\frac{\vert \eta(\alpha)\vert^2}{\sin (\gamma - \alpha)} \,, 
\end{equation}
which plays an important role in the sequel.
In the period  interval $[-\pi, \pi)$ and for $0\leq \gamma <\pi$,  this function has support in the interval $(\gamma-\pi,\gamma)$, as does $\eta$, and it is normalised in the sense that 
\begin{equation}
\label{normEq}
\int_{\gamma-\pi}^{\gamma+\pi}\ud \alpha \, E_{\eta;\gamma}(\alpha)= \int_{\gamma-\pi}^{\gamma}\ud \alpha \, E_{\eta;\gamma}(\alpha)=1\, .
\end{equation}
 Thus it can be considered a probability distribution on the interval $[\gamma-\pi, \gamma]$ (or $[-\pi,\pi]$), and the average value of a function $f(\alpha)$ on the same interval will be denoted by
 \begin{equation}
\label{averEf}
\lg f\rg_{E_{\eta;\gamma}}:= \int_{-\pi}^{+\pi}\ud \alpha \,f(\alpha) E_{\eta;\gamma}(\alpha)= \int_{\gamma-\pi}^{\gamma}\ud \alpha \,f(\alpha) E_{\eta;\gamma}(\alpha)\,. 
\end{equation}

The application of \eqref{x0039} and \eqref{kerAf} to the quantisation of functions  which only depend on the angle is straightforward and leads to the following result. 
\begin{prop}
For $f(p,q)= u(q)$ with $u(q+2\pi)= u(q)$,  $A_{u}$ is the multiplication operator
\begin{equation}
\label{AuQ}
(A_{u}\psi) (\alpha) =  \left(E_{\eta;\gamma}\ast u\right) (\alpha)\, \psi(\alpha)\,, 
\end{equation}
where the \textit{periodic convolution} product on the circle is defined by 
\begin{equation}
\label{convcirc}
(E\ast u)(\alpha)= \int_0^{2\pi} \ud q\, E(\alpha -q) \, u(q)\, .
\end{equation}
\end{prop}
Moreover, since the function $E_{\eta;\gamma}$ is a probability distribution on a period interval, a standard result of Analysis \cite{schilling06} on convolution allows us to state the following.
\begin{prop}
\label{bouncont}
If the $2\pi$-periodic function $u$ is bounded on a period interval, then the $2\pi$-periodic convolution $(E_{\eta;\gamma}\ast u)$ is bounded and continuous. 
\end{prop}
\subsubsection*{An elementary example: the Fourier exponential}
The operator $A_{\De_n}$ associated with the Fourier exponential $\De_{n}(\alpha)=e^{\ii n\alpha}$ is given by \eqref{AuQ}. The convolution $E_{\eta;\gamma}\ast\De_n$ takes the form
\begin{equation}
\left( E_{\eta;\gamma}\ast\De_n\right)(\alpha) =
\int_{\alpha-\gamma}^{\alpha+\pi-\gamma}\ud q\,
E_{\eta;\gamma}(\alpha-q)\,e^{\ii nq}\, .
\end{equation}
The change of variables $q\rightarrow \alpha-q$ yields the multiplication operator
\begin{equation}
\left( E_{\eta;\gamma}\ast\De_n\right)(\alpha) =
\left(\int^{\gamma}_{\gamma-\pi}\ud q\,E_{\eta;\gamma}(q)\,
e^{-\ii nq}\right) \, e^{\ii n\alpha}
=\text{const.} \, e^{\ii n\alpha}\, . 
\end{equation}
A suitable choice of the fiducial vector allows one to obtain $\text{const.} = 1$. Thus the quantum versions of  simple trigonometric functions, like $\sin\alpha$, $\cos\alpha$, are multiplication operators defined by these classical functions, as is the case with many other approaches \cite{levyleblond76}. 
\subsection{Quantisation of a function of $p$}\label{quantfuncp}
\subsubsection{Momentum $p$}
For the momentum $f(p,q)=p$, the expression (\ref{kerAf}) becomes
\begin{equation}
\begin{aligned}
\left(A_{p}\psi\right)\left(\alpha\right)=\dfrac{-\ii\pi}{\kappa c_{\eta}}\int_{\mathbb{S}^{1}}\textrm{d}q\,\eta(\alpha-q)\int_{\mathbb{S}^{1}}\textrm{d}\alpha^{\prime}
\,\overline{\eta(\alpha^{\prime}-q)}e^{2\ii \lambda S_{\zeta}\left(\alpha,\alpha^{\prime},q\right)}\psi\left(\alpha^{\prime}\right)\dfrac{\partial \delta\left(2\kappa S_{\gamma}\left(\alpha,\alpha^{\prime},q\right)\right)}{\partial S_{\gamma}\left(\alpha,\alpha^{\prime},q\right)}\, .
\end{aligned}
\end{equation}
Taking into account the support of $\eta$ from theorem \ref{theorem1}, the above integral reduces to
\begin{equation}
\begin{aligned}
\label{eqAp}
\left(A_{p}\psi\right)\left(\alpha\right)=\dfrac{-\textrm{i}\pi}{\kappa^2 c_{\eta}}\int_{\mathbb{S}^{1}}\textrm{d}q\,
\dfrac{\eta(\alpha-q)}{\sin\left(q+\gamma-\alpha\right)}
\int_{\mathbb{S}^{1}}\textrm{d}\alpha^{\prime}\,\overline{\eta(\alpha^{\prime}-q)}e^{2\textrm{i}\lambda S_{\zeta}\left(\alpha,\alpha^{\prime},q\right)}
\\ \times
\left(\dfrac{\partial S_{\gamma}\left(\alpha,\alpha^{\prime},q\right)}{\partial\alpha^{\prime}}\right)^{-1}\psi\left(\alpha^{\prime}\right)
\dfrac{\partial}{\partial\alpha^{\prime}}\delta\left(\alpha-\alpha^{\prime}\right)\, .
\end{aligned}
\end{equation}
Integrating (\ref{eqAp}) with respect to $\alpha^{\prime}$, using again the conditions on $\text{supp}\,\eta$  and making the change of variables $q^{\prime}=\alpha-q$, gives
\begin{equation}
\label{eqn:longexpress1}
\begin{aligned}
\left(A_{p}\psi\right)\left(\alpha\right)=-\textrm{i}\dfrac{c_{2}(\eta,\gamma)}{\kappa c_{1}(\eta,\gamma)}\partial_{\alpha}\psi\left(\alpha\right)
-
\textrm{i}\dfrac{1}{\kappa c_{1}(\eta,\gamma)} \left(
\int_{\mathbb{S}^{1}}\textrm{d}q\,\cos\left(\gamma-q\right)   f_{0;3}(q)
\right.
\\
+
\int_{\mathbb{S}^{1}}\textrm{d}q\,f_{1;2}(q)
-\left.\textrm{i}\lambda\int_{\mathbb{S}^{1}}\textrm{d}q\,\sin\left(\zeta-q\right) f_{0;2}(q)\right)
\psi\left(\alpha\right)\, ,
\end{aligned}
\end{equation}
where the functions $f_{j;m}(q)$ are defined in \eqref{functionsfjm}. The quantity  $\int_{\mathbb{S}^{1}}\textrm{d}q\,\cos\left(\gamma-q\right)   f_{0;3}(q)+\int_{\mathbb{S}^{1}}\textrm{d}q\,f_{1;2}(q)$ is purely imaginary, so it vanishes for real $\eta(\alpha)$. The expression (\ref{eqn:longexpress1}) takes the form
\begin{equation}
\label{eq:quantp}
\left(A_{p}\psi\right)\left(\alpha\right)=\left(-\textrm{i}\dfrac{c_{2}(\eta,\gamma)}{\kappa c_{1}(\eta,\gamma)}\dfrac{\partial}{\partial\alpha}-\lambda a \right) \psi\left(\alpha\right)\, ,
\end{equation}
where the constant $a$  is 
\begin{equation}
\label{aprime}
a =\dfrac{1}{\kappa c_{1}(\eta,\gamma)} \int_{\mathbb{S}^1} \ud q\,\sin(\zeta-q) f_{0;2}(q)\, . 
\end{equation}
We note that with the admissible choice
\begin{equation}
\label{kappachoice}
\kappa= \dfrac{c_2(\eta,\gamma)}{c_1(\eta,\gamma)}
\end{equation}
 one gets, up to the addition of an irrelevant constant, the self-adjoint angular momentum operator $-\ii\partial/\partial\alpha$, with spectrum $n\in \Z$ and Fourier exponentials $e^{\ii n\alpha}$ as corresponding eigenfunctions. Note that to the same effect one can choose $\kappa=1$ and $\eta$ in such a way that $\dfrac{c_2(\eta,\gamma)}{c_1(\eta,\gamma)}$ is arbitrarily close to $1$.

\subsection{Quantisation of simple separable functions}
Many, if not all physically relevant Hamiltonians for one-dimensional systems can be written in the form $H= u_2(q) p^2 + u_1(q) p + u_0(q)$. Thus, it is useful to give the expressions of their quantum counterparts obtained by means of our method.  We first define a set of functions which helps to express these operators  in a simple way. 
\begin{defi}
Given $f_{j;m}(q)$ defined in \eqref{functionsfjm}, the set of periodic functions $B_{j}(q)$, $j=1,..,5$, is defined as
\begin{equation}
\begin{aligned}
B_{1}(q)=-\frac{1}{\kappa^{2} c_{1}\left(\eta,\gamma\right)}f_{0;3}(q)\, ,
\end{aligned}
\end{equation}
\begin{equation}
\begin{aligned}
B_{2}(q)=\ii\frac{2\lambda}{\kappa^{2} c_{1}\left(\eta,\gamma\right)}
f_{0;3}(q)\sin(\zeta-q)\, ,
\end{aligned}
\end{equation}
\begin{equation}
\begin{aligned}
B_{3}(q)=\frac{1}{\kappa^{2} c_{1}\left(\eta,\gamma\right)}\left[ 
-f_{2;3}(q)-3f_{1;4}(q)\cos(\gamma-q)
-3f_{0;5}(q)(\cos(\gamma-q))^{2}\right. \\
-f_{0;3}(q)+\lambda^{2}f_{0;3}(q)\left(\sin(\zeta-q)\right)^{2}+\ii2\lambda f_{1;3}(q)\sin(\zeta-q)\\
\left.-\ii\lambda f_{0;3}(q)\cos(\zeta-q) 
 +\ii\lambda 3f_{0;4}(q)\sin(\zeta-q)\cos(\gamma-q)\right]\, ,
\end{aligned}
\end{equation}

\begin{equation}
\label{equFprimeeta4}
B_4(q)=\frac{ c_{2}(\eta,\gamma)}{\kappa c_{1}(\eta,\gamma)}f_{0;2}(q)\, ,
\end{equation}

\begin{equation}
\label{equFprimeeta5}
B_5(q)=\frac{\lambda }{\kappa c_{1}(\eta,\gamma)}
\sin\left(\zeta-q\right) f_{0;2}(q)\, ,
\end{equation}

where the constants $\zeta$, $\lambda$, $\eta$ and $\gamma$ are given by Definition \ref{defics}.

\end{defi}

\subsubsection{quantisation of momentum squared}

For $f(q,p)=p^2$, which is the classical kinetic term up to a multiplicative constant, one has
\begin{equation}
\label{momentumsquare}
\left( A_{p^2}\psi\right) (\alpha)=\left[ b_1 \partial_{\alpha}^{2}
 +b_2 \partial_{\alpha}
 +b_3 \right]  \psi(\alpha)\,,
\end{equation}
where
\begin{equation}
b_j=\int_{\mathbb{S}^{1}}\ud q\, B_{j}(q)\, ,\quad j=1,2,3\, .
\end{equation}
 %b_1+b_2\partial_{\alpha}+b_3
\subsubsection{quantisation of other simple product functions}
For the functions $pu(q)$ and $p^{2}u(q)$ one obtains

\begin{equation}
\begin{aligned}
\label{eqApu}
\left(A_{pu(q)}\psi\right)(\alpha)=-\ii \left(u\ast B_4\right)(\alpha)\partial_{\alpha}\psi(\alpha)
- \left(u\ast B_5\right)(\alpha)\psi(\alpha)\, ,
\end{aligned}
\end{equation}

and 

\begin{equation}
\label{momentumsquareU}
\left( A_{p^2 u(q)}\psi\right) (\alpha)=\left[ \left(u\ast B_{1}\right)(\alpha) \partial_{\alpha}^{2}
 +\left(u\ast B_{2}\right)(\alpha)  \partial_{\alpha}
 +\left(u\ast B_{3}\right)(\alpha) \right]  \psi(\alpha)\,.
\end{equation}
%---------------------------------------------------------------------------------------
\section{Computation of semi-classical portraits}\label{sec:lower}

Following (\ref{lower1111}), the classical portrait $\check{f}(p,q)$ of the operator $A_{f}$ is given by
\begin{equation}
\label{lower1}
\check{f}(p,q)=\bra{\eta_{p,q}}A_{f}\ket{\eta_{p,q}}=
\int_{\R\times\mathbb{S}^{1}}\frac{dp^{\prime}dq^{\prime}}{c_{\eta}}f\left(p^{\prime},q^{\prime}\right)
\left|\left\langle \eta_{p^{\prime},q^{\prime}}|\eta_{p,q}\right\rangle \right|^{2}\, .
\end{equation}
Using the functions from the definition \ref{defi1} we obtain for \eqref{lower1} the involved integral expression
\begin{equation}
\label{lower2}
\begin{aligned}
\check{f}(p,q)=\dfrac{1}{c_{\eta}}
\int_{\mathbb{S}^{1}}\ud q^{\prime}
\int_{\mathbb{S}^{1}}\ud \alpha\,
\overline{\eta(\alpha-q^{\prime})}\eta(\alpha-q)
\int_{\mathbb{S}^{1}}\ud \alpha^{\prime}
\eta(\alpha^{\prime}-q^{\prime})\overline{\eta(\alpha^{\prime}-q)}\\ \times\,
e^{\ii 2\kappa S_{\gamma}(\alpha^{\prime},\alpha,q) p}
e^{\ii 2\lambda S_{\zeta}(\alpha,\alpha^{\prime},q^{\prime}) }
e^{\ii 2\lambda S_{\zeta}(\alpha^{\prime},\alpha,q) }
\int_{-\infty}^{+\infty}\ud p^{\prime}\, e^{\ii 2\kappa S_{\gamma}(\alpha,\alpha^{\prime},q^{\prime}) p^{\prime}}\, f(p^{\prime},q^{\prime})\, .
\end{aligned}
\end{equation}
Two simple applications are examined below. 
\subsection{Lower symbol for a function of $q$}
In the case where $f(p,q)=u(q)$, after integrating with respect to $p^{\prime}$ and restricting to $\text{supp}\,\eta$, the expression (\ref{lower2}) becomes
\begin{equation}
\label{eql880}
\begin{aligned}
\check{u}(q)=\frac{2\pi}{\kappa c_{\eta}}\int_{\mathbb{S}^{1}}\ud q^{\prime}\int_{\mathbb{S}^{1}}\ud\alpha\,
\overline{\eta(\alpha-q^{\prime})}\eta(\alpha-q)
\int_{\mathbb{S}^{1}}\ud\alpha^{\prime}\,
\eta(\alpha^{\prime}-q^{\prime})\overline{\eta(\alpha^{\prime}-q)}\\
\times 
e^{\ii 2\kappa S_{\gamma}(\alpha^{\prime},\alpha,q) p}
e^{\ii 2\lambda S_{\zeta}(\alpha,\alpha^{\prime},q^{\prime}) }
e^{\ii 2\lambda S_{\zeta}(\alpha^{\prime},\alpha,q) }
\dfrac{u(q^{\prime})\delta\left(\alpha^{\prime}-\alpha\right)}
{\sin(\gamma-(\alpha-q^{\prime}))}\, .
\end{aligned}
\end{equation}
Integrating (\ref{eql880}) with respect to $\alpha^{\prime}$ gives
\begin{equation}
\check{u}(q)=\frac{2\pi}{\kappa c_{\eta}}\int_{\mathbb{S}^{1}}\ud\alpha\,\left|\eta(\alpha-q)
\right|^{2}\int_{\mathbb{S}^{1}}\ud q^{\prime}\,\dfrac{u(q^{\prime})\left|\eta(\alpha-q^{\prime})\right|^{2}}
{\sin(\gamma-(\alpha-q^{\prime}))}\, .
\end{equation}
Finally, using $E_{\eta;\gamma}(\alpha)$ from (\ref{AuQ}) it is easy to see that
\begin{equation}
\label{loweruofq}
\check{u}(q)=\int_{\mathbb{S}^{1}}\ud\alpha\,\left|\eta(\alpha-q)\right|^{2}\,
\left( E_{\eta;\gamma}\ast u\right) (\alpha)\, .
\end{equation}
The same result is obtained by calculating $\check{u}=\int_{\mathbb{S}^{1}}\ud\alpha\,\overline{\eta_{p,q}(\alpha)}\left(A_{u} \eta_{p,q}\right)(\alpha)$ with the expression \eqref{AuQ} for  the operator $A_u$. Considering the function $\widetilde{\eta}(\alpha)=\left|\eta(-\alpha)\right|^{2}$, the expression \eqref{loweruofq} can be written as the convolution
\begin{equation}
\label{loweruofqconv}
\check{u}(q)=\left[\widetilde{\eta}\ast\left( E_{\eta;\gamma}\ast u\right) \right](q)\, .
\end{equation}
\subsection{Lower symbol for the momentum $p$}
For $f(p,q)=p$, after integrating with respect to $p^{\prime}$  the expression (\ref{lower2}) becomes
\begin{equation}
\begin{aligned}
\check{p}=\frac{-\ii\pi}{\kappa c_{\eta}}\int_{\mathbb{S}^{1}}\ud q^{\prime}\int_{\mathbb{S}^{1}}\ud\alpha
\,\overline{\eta(\alpha-q^{\prime})}\eta(\alpha-q)
\int_{\mathbb{S}^{1}}\ud\alpha^{\prime}\,
\eta(\alpha^{\prime}-q^{\prime})\overline{\eta(\alpha^{\prime}-q)}\\
\times 
e^{\ii 2\kappa S_{\gamma}(\alpha^{\prime},\alpha,q) p}
e^{\ii 2\lambda S_{\zeta}(\alpha,\alpha^{\prime},q^{\prime}) }
e^{\ii 2\lambda S_{\zeta}(\alpha^{\prime},\alpha,q) }
\dfrac{\partial}{\partial S_{\gamma}(\alpha^{\prime},\alpha,q^{\prime})}
\delta\left(2\kappa S_{\gamma}(\alpha,\alpha^{\prime},q^{\prime})\right)\, .
\end{aligned}
\end{equation}
Assuming that $\eta$ is a real function and proceeding as in the previous case, one arrives at
\begin{equation}
\label{eq:semiclassp}
\begin{aligned}
\check{p}=\frac{c_{2}(\eta,\gamma)}{c_{1}(\eta,\gamma)}c_{-1}(\eta,\gamma)p+\lambda\frac{c_{2}(\eta,\gamma)}{\kappa c_{1}(\eta,\gamma)} c_{-1}(\eta,\zeta)-\lambda a\, .
\end{aligned}
\end{equation}
where the constant $a$ is given by \eqref{aprime}. The same result is obtained by calculating $\check{p}=\int_{\mathbb{S}^{1}}\ud\alpha\,\overline{\eta_{p,q}(\alpha)}\left(A_{p} \eta_{p,q}\right)(\alpha)$ with the expression \eqref{eq:quantp} for  the operator $A_p$.
%---------------------------------------------------------------------------------------
\section{Angle operator }\label{sec:spectral}
\subsection{The angle operator: general properties}
%----------
Before presenting a key result of our paper, we need to define a $2\pi$-periodic function associated to the function $E_{\eta;\gamma}$ introduced in \eqref{Eal}. 
\begin{defi}
\label{perfunF}
Given $\gamma \in [0,\pi)$, the $2\pi$-periodic function $\mathcal{F}_{\eta;\gamma}$ is defined for $\alpha\in[0,2\pi)$ by
\begin{equation}
\mathcal{F}_{\eta;\gamma}(\alpha)=
2\pi\times\left\lbrace
\begin{array}{ccc}
\int_{\alpha}^{\gamma}\ud q\, E_{\eta;\gamma}(q)\,,
& &  0\leqslant\alpha<\gamma
\\ \\
0 \,,& & \gamma\leqslant\alpha\leqslant\pi+\gamma
\\ \\
-\int^{\alpha-2\pi}_{\gamma-\pi}\ud q\, E_{\eta;\gamma}(q)= -1 +  \int^{\gamma}_{\alpha-2\pi}\ud q\, E_{\eta;\gamma}(q)\,,
& & \pi+\gamma<\alpha<2\pi
\end{array}
\right.\, .
\end{equation}
\end{defi} %------------
The \textit{second fundamental theorem of calculus} provides the following  property of the function $\mathcal{F}_{\eta;\gamma}$.
\begin{prop}
Suppose that the function $E_{\eta;\gamma}(\alpha)$ is continuous.  Then the function $\mathcal{F}_{\eta;\gamma}$ is  piecewise continuous and differentiable for $\alpha \neq 2k\pi$. For $\alpha \in [\gamma-\pi,\gamma]\setminus\{0\}$, in  the period interval $[-\pi,\pi)$, its derivative is given by
 \begin{equation}
\label{derFalp}
\frac{\ud}{\ud \alpha} \mathcal{F}_{\eta;\gamma}(\alpha)= - 2\pi 
E_{\eta;\gamma}(\alpha)
\, .
\end{equation}
At the discontinuity point $\alpha = 0$, the jump of $\mathcal{F}_{\eta;\gamma}$ is $2\pi$, and its  derivative is $2\pi \delta(0)$. 
\end{prop}
\begin{theo}
For the discontinuous $2\pi$-periodic angle function $\Da(\alpha)$ defined by
\begin{equation}
\label{anglefunction}
\Da(\alpha) = \alpha \quad \mbox{for} \quad \alpha\in[0,2\pi)\, ,
\end{equation}
the angle operator is the bounded self-adjoint multiplication operator
\begin{equation}
\label{eq:angopalpha}
\left(A_{\Da}\psi\right)(\alpha)=\left(E_{\eta;\gamma}\ast\Da\right)(\alpha)\psi(\alpha)\, ,
\end{equation} 
where the convolution $E_{\eta;\gamma}\ast\Da$  is given by the $2\pi$-periodic bounded continuous function
\begin{equation}
\label{analyticsoectrum}
\left(E_{\eta;\gamma}\ast \Da\right) (\alpha)=\Da(\alpha) +\mathcal{F}_{\eta;\gamma}(\alpha) - 
\int^{\gamma}_{\gamma -\pi}\ud q\, q\, E_{\eta;\gamma}(q)
\, .
\end{equation}
We note that the lower $2\pi$  jumps of the angle function $\Da$ at $\alpha= 2k\pi$ are exactly canceled by the upper jumps of the function $\mathcal{F}_{\eta;\gamma}$ at the same points. 
\end{theo}
%---
%---
\begin{proof}
Taking into account the conditions on  $\text{supp}\,\eta$ , the convolution $E_{\eta;\gamma}\ast\Da$ becomes
\begin{equation}
\label{angle1op1}
\left(E_{\eta;\gamma}\ast \Da\right) (\alpha)=   \int_{\alpha-\gamma}^{\alpha-\gamma +\pi}\ud q\, E_{\eta;\gamma}(\alpha-q)\,\Da (q) \, .
\end{equation}
The convolution \eqref{angle1op1} restricted to the interval $[0,\gamma)$ is given by
\begin{equation}
\begin{aligned}
\left. \left(E_{\eta;\gamma}\ast \Da\right) \right|_{[0,\gamma)} (\alpha) = 
\int_{\alpha-\gamma}^{0}\ud q\,(q+2\pi)\, E_{\eta;\gamma}(\alpha-q)\\
+ 
\int_{0}^{\alpha-\gamma +\pi}\ud q\, q\,E_{\eta;\gamma}(\alpha-q)\, .
\end{aligned}
\end{equation}
After the change of variables $q\rightarrow\alpha-q$, one has
\begin{equation}
\label{inter1}
\begin{aligned}
\left. \left(E_{\eta;\gamma}\ast \Da\right) \right|_{[0,\gamma)}  (\alpha)=  
\alpha-\int^{\gamma}_{\gamma -\pi}\ud q\, q\, E_{\eta;\gamma}(q)
+ 2\pi\int^{\gamma}_{\alpha}\ud q\, E_{\eta;\gamma}(q)\, .
\end{aligned}
\end{equation}
For the other intervals one has
\begin{equation}
\label{inter2}
\begin{aligned}
\left. \left(E_{\eta;\gamma}\ast \Da\right) \right|_{[\gamma,\pi+\gamma]} (\alpha)=  
\alpha-\int^{\gamma}_{\gamma -\pi}\ud q\, E_{\eta;\gamma}(q)\,q \, ,
\end{aligned}
\end{equation}
\begin{equation}
\label{inter3}
\begin{aligned}
\left. \left(E_{\eta;\gamma}\ast \Da\right) \right|_{(\pi+\gamma,2\pi)} (\alpha)=  
\alpha-\int^{\gamma}_{\gamma-\pi}\ud q\, q\, E_{\eta;\gamma}(q)
- 2\pi\int^{\alpha-2\pi}_{\gamma-\pi}\ud q\, E_{\eta;\gamma}(q)\, .
\end{aligned}
\end{equation}
Combining the equations \eqref{inter1}, \eqref{inter2}, \eqref{inter3}, \eqref{averEf}, and the definition \ref{perfunF}, we arrive at the expression \eqref{analyticsoectrum} of $A_{\Da}$ as a multiplication operator on $L^2(\mathbb{S}^1, \ud\alpha)$. Since the angle function is bounded, then from Prop. \ref{bouncont} the convolution function $E_{\eta;\gamma}\ast\Da$ is continuous and bounded. Since the latter is also real,  the corresponding multiplication operator is bounded self-adjoint. 
\end{proof}

\subsection{Angle operator: analytic and numerical results}
A specific section $\sigma$ is now used (which implies a choice of $\bm{\kappa},\bm{\lambda}\in\R^{2}$). To simplify, we put $\lambda=0$, $\gamma=\pi/2$ and select a real, even, smooth  $\eta(\alpha)$. Hence, its support is a subset of the right half-circle  $-\dfrac{\pi}{2} < \alpha < \dfrac{\pi}{2}$. We already proved that the spectrum of the angle operator $A_{\Da}$ is purely continuous and given by the range of the function 
\begin{equation}
\label{speclambdaonehalf}
\left(E_{\eta;\frac{\pi}{2}}\ast \Da\right) (\alpha)=\Da(\alpha)-\int^{\frac{\pi}{2}}_{-\frac{\pi}{2}}\ud q\, q\, E_{\eta;\frac{\pi}{2}}(q)
+\mathcal{F}_{\eta;\frac{\pi}{2}}(\alpha)\, .
\end{equation}
In order to study the relation between the localisation of $\eta$ and the spectrum of $A_{\Da}$, we pick the  familiar smooth and compactly supported  test functions for distributions, namely,
\begin{equation}
\label{testeps}
\omega_{\epsilon}(x) = \left\lbrace \begin{array}{cc}
  \exp\left(-\dfrac{\epsilon}{1-x^2}\right)    & 0\leq\vert x\vert < 1\, ,    \\
   0   &   \vert x\vert \geq 1\,, 
\end{array}\right.
\end{equation}
where the parameter $\epsilon \geq 0$ determines the rate of decrease of $\omega_{\epsilon}$. We also note that $0\leq \omega_{\epsilon}(x)\leq e^{-\epsilon}$. 
Now we choose as fiducial vectors the family of $2\pi$-periodic  smooth even functions which have  support $[-\delta,\delta]\subset(-\pi/2,\pi/2)$ and which are  parametrized by $\epsilon>0$ and $0<\delta<\pi/2$,
\begin{equation}
\label{eq:familydelta}
\eta(\alpha)\equiv \eta^{(\epsilon,\delta)}(\alpha)= \frac{1}{\sqrt{\delta e_{2\epsilon}}}\, \omega_{\epsilon}\left(\frac{\alpha}{\delta}\right)\, . 
\end{equation}
The normalisation factor involves the integral 
\begin{equation}
\label{eeps}
e_{\epsilon}:= \int_{-1}^{1}\ud  x \, \omega_{\epsilon}(x)\, . 
\end{equation}
As a function of $\epsilon\in [0, +\infty)$, $e_{\epsilon}$ decreases uniformly from $2$ to $0$. With these definitions, the integrals $c_{\nu}\left(\eta,\frac{\pi}{2}\right)$ defined by \eqref{cetanu} assume the simple form
\begin{equation}
\label{vetanuom}
c_{\nu}\left(\eta^{(\epsilon,\delta)},\frac{\pi}{2}\right)= \frac{1}{e_{2\epsilon}}\int_{-1}^1 \ud x\, \frac{\omega_{2\epsilon}(x)}{(\cos \delta x)^{\nu}}\equiv \frac{e_{2\epsilon}(\delta, \nu)}{e_{2\epsilon}}\,,  
\end{equation}
where
\begin{equation}
\label{eepsdelnu}
e_{\epsilon}(\delta, \nu) := \int_{-1}^1 \ud x\, \frac{\omega_{\epsilon}(x)}{(\cos \delta x)^{\nu}}\, , \quad 0\leq \delta < \frac{\pi}{2}\, . 
\end{equation}
Graphs of the function $ \eta^{(\epsilon,\delta)}(\alpha)$ for a few values of parameters $\delta$ and $\epsilon$  are  shown in Figures \ref{fig:etafunction1} and \ref{fig:etafunction2}. They give an idea of its localisation properties.   

With the above notations 
\begin{equation}
\label{E1}
E_{\eta;\frac{\pi}{2}}(\alpha)= E_{\eta^{(\epsilon,\delta)};\frac{\pi}{2}}(\alpha)= \dfrac{1}{\delta e_{2\epsilon}(\delta,1)}\dfrac{\omega_{2\epsilon}\left(\dfrac{\alpha}{\delta}\right)}{\cos\alpha} \, .
\end{equation}
\begin{figure}
\centering
\begin{subfigure}[b]{0.4\textwidth}
\includegraphics[width=\textwidth]{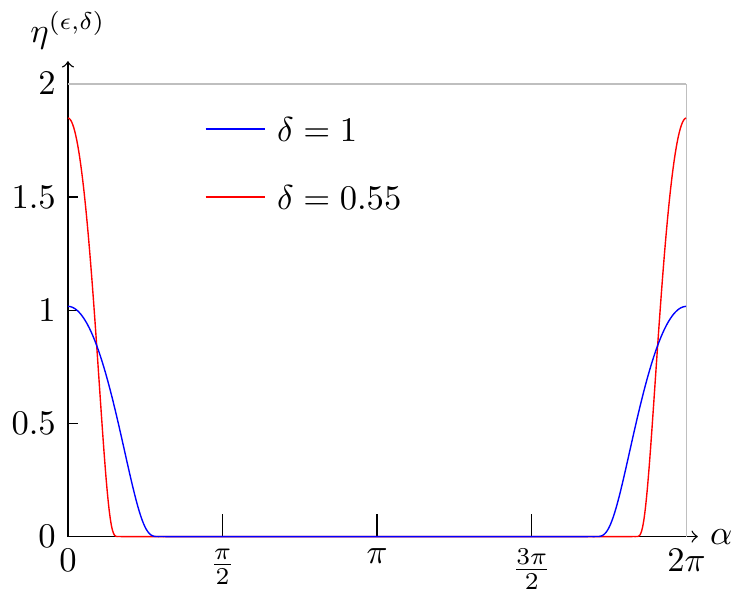}
\caption{$\epsilon=1$}
\label{fig:etafunction1}
\end{subfigure}
\begin{subfigure}[b]{0.4\textwidth}
\includegraphics[width=\textwidth]{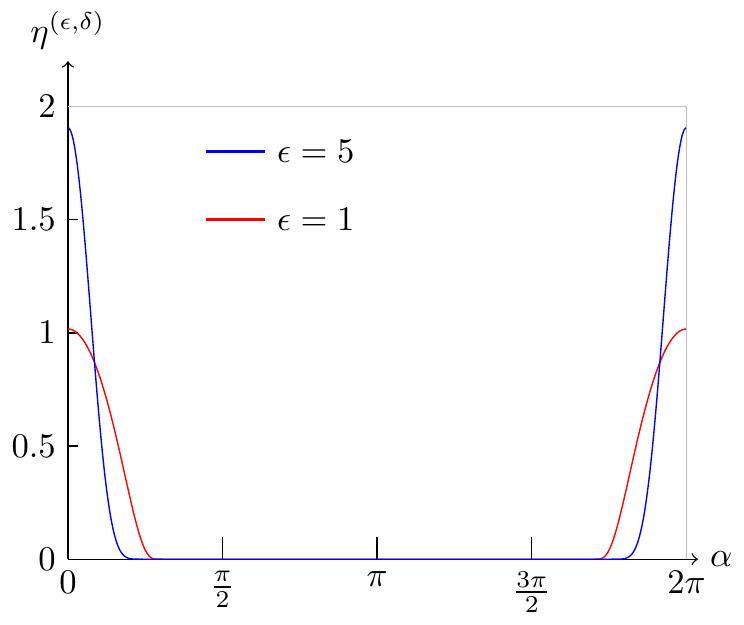}
\caption{$\delta=1$}
\label{fig:etafunction2}
\end{subfigure}
\caption{The parameters $\epsilon$ and $\delta$ control the localisation about $\alpha=0 \, \mathrm{mod}\,2\pi$ of the fiducial function (i.e., ``wavelet") $\eta^{(\epsilon,\delta)}(\alpha)$. In the figure \ref{fig:etafunction1} the localisation improves as $\delta$ becomes smaller. In the figure \ref{fig:etafunction2} the localisation improves as $\epsilon$ becomes larger.}\label{fig:etafunction}
\end{figure}
Taking into account the dependence of $\text{supp}\,\eta$ on $\delta$,
the convolution \eqref{speclambdaonehalf} for $\alpha\in[-\pi,\pi)$ is given by
\begin{equation}
\label{eq:spct}
\left( E_{\eta^{(\epsilon,\delta)};\frac{\pi}{2}}\ast\Da\right)(\alpha) = \alpha + 2\pi\times
\left\lbrace
\begin{array}{ccc}
\int_{\alpha}^{\delta}\ud x\, E_{\eta^{(\epsilon,\delta)};\frac{\pi}{2}}(x)& & 0\leqslant\alpha <\delta\\
0 & & \delta\leqslant\alpha\leqslant 2\pi-\delta\\
-\int_{-\delta}^{\alpha-2\pi}\ud x\, E_{\eta^{(\epsilon,\delta)};\frac{\pi}{2}}(x)& & 2\pi-\delta<\alpha< 2\pi\\
\end{array}
\right. \, ,
\end{equation}
We note that $\int^{\frac{\pi}{2}}_{-\frac{\pi}{2}}\ud q\, q\, E_{\eta;\frac{\pi}{2}}(q)$ vanishes due to the even parity of the fiducial function. 
\begin{figure}
    \centering
    \begin{subfigure}[b]{0.6\textwidth}
		\includegraphics[width=\textwidth]{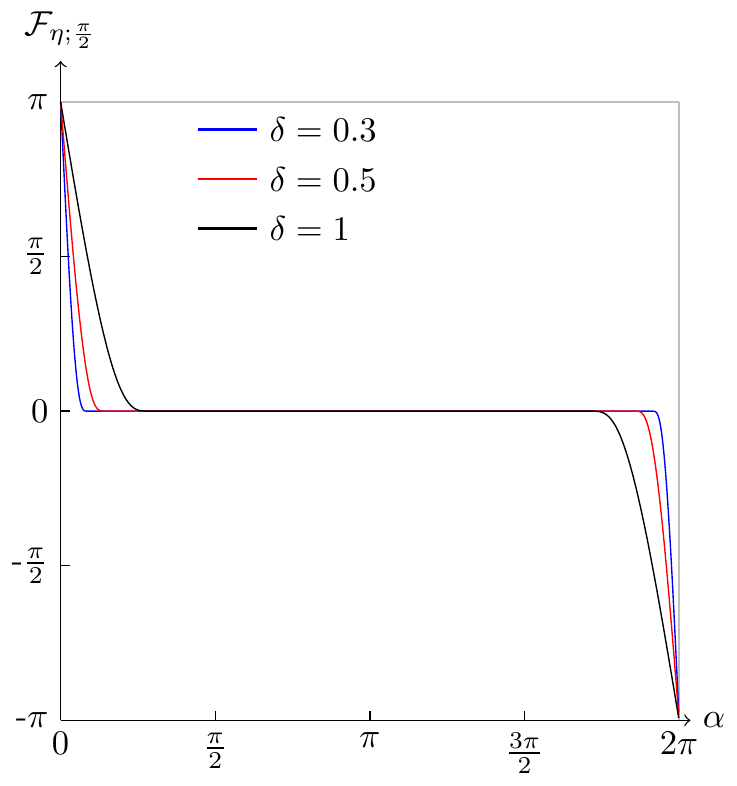}
		\caption{Function $\mathcal{F}_{\eta^{(\epsilon,\delta)};\frac{\pi}{2}}(\alpha)$ for $\epsilon=1$.}
		\label{fig:functionF}
    \end{subfigure} 
    \begin{subfigure}[b]{0.6\textwidth}
        \includegraphics[width=\textwidth]{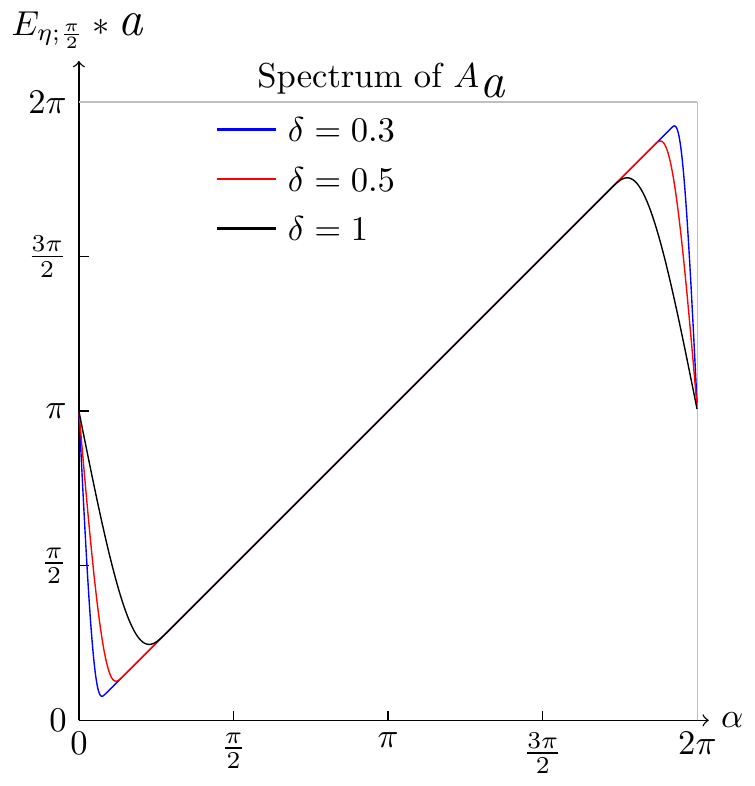}
        \caption{Spectrum of $A_{\Da}$ for $\epsilon=1$.}
        \label{fig:angleopdelta1}
    \end{subfigure}
   \caption{The continuity of the convolution $\left( E_{\eta^{(\epsilon,\delta)};\frac{\pi}{2}}\ast\Da\right)(\alpha)$ is a result of the regularization of $\Da(\alpha)$ by $\mathcal{F}_{\eta^{(\epsilon,\delta)};\frac{\pi}{2}}(\alpha)$. }\label{fig:angleopdelta}
\end{figure}
The expression (\ref{eq:spct}) above is a periodic bounded function, and we see that it coincides with the angle function $\Da$ inside $[\delta,2\pi-\delta]$. Outside the interval $[\delta,2\pi-\delta]$ the function $\mathcal{F}_{\eta^{(\epsilon,\delta)};\frac{\pi}{2}}(\alpha)$ regularizes $\Da$, therefore the convolution (\ref{eq:spct}) becomes a continuous function. The behaviour of the quantum angle viewed as  the multiplication operator \eqref{eq:spct}, for different values of $\delta$,  is depicted in the figure \ref{fig:angleopdelta}. As $\delta$ approaches zero (or $\epsilon$ becomes very large), the convolution $E_{\eta_{\delta};\frac{\pi}{2}}\ast\Da$ can be made arbitrarily close to the angle function $\Da$.
We notice from the figure that the spectrum $\sigma\left(A_{\Da}\right)$ of the angle operator is continuous, as expected from the smoothness of the convolution, and is, for a given $\delta$ and $\epsilon$, a closed interval strictly included in the interval $[0,2\pi]$, i.e.,
\begin{equation}
\label{Aspectrum}
\sigma\left(A_{\Da}\right) = [\pi -m(\epsilon,\delta), \pi + m(\epsilon,\delta)]\, , \quad 0< m(\epsilon,\delta)< \pi\, , 
\end{equation} 
with $m(\epsilon,\delta)\to \pi$ as $\delta \to 0$ or $\epsilon\to\infty$, i.e., the spectrum goes to $[0,2\pi)$. For a fixed value of $\epsilon$, the real number $\pi -m(\epsilon,\delta)$ corresponds to the positive root $\alpha\in(0,\delta)$ of the equation
\begin{equation}
\omega_{2\epsilon}\left(\frac{\alpha}{\delta}\right)=\frac{\delta e_{2\epsilon}(\delta,1)}{2\pi }\cos\alpha\, .
\end{equation}

It is interesting to compare the function \eqref{eq:spct} with the semiclassical portrait of $A_{\Da}$. Considering the function  $\widetilde{\eta}^{(\epsilon,\delta)}(\alpha)=\eta^{(\epsilon,\delta)}(-\alpha)$, using the expression \eqref{loweruofqconv} the semiclassical portrait of $A_{\Da}$ can be written as
\begin{equation}
\check{q}(q)=
\left[\left(\widetilde{\eta}^{(\epsilon,\delta)}\right)^2\ast\left( E_{\eta^{(\epsilon,\delta)};\frac{\pi}{2}}\ast\Da\right)\right](q) .
\end{equation}
Taking into account the support of $\eta^{(\epsilon,\delta)}$ one has explicitly
\begin{equation}
\label{eq:semiclassangle}
\check{q}=
\left\lbrace 
\begin{array}{ccc}
\int_{q+2\pi-\delta}^{2\pi}\text{d}\alpha\,\left(\eta^{(\epsilon,\delta)}(\alpha-q-2\pi)\right)^2\,\left(E_{\eta^{(\epsilon,\delta)};\frac{\pi}{2}}\ast \Da\right)(\alpha)
&&\\
+\int_{0}^{q+\delta}\text{d}\alpha\,\left(\eta^{(\epsilon,\delta)}(\alpha-q)\right)^2\,\left(E_{\eta^{(\epsilon,\delta)};\frac{\pi}{2}}\ast \Da\right)(\alpha)
&& 0\leq q<\delta\\
\quad&&\\
\int_{q-\delta}^{q+\delta}\text{d}\alpha\,\left(\eta^{(\epsilon,\delta)}(\alpha-q)\right)^2\,\left(E_{\eta^{(\epsilon,\delta)};\frac{\pi}{2}}\ast \Da\right)(\alpha)&&\delta\leq q<2\pi-\delta\\
\quad&&\\
\int_{0}^{q-2\pi+\delta}\text{d}\alpha\,\left(\eta^{(\epsilon,\delta)}(\alpha-q+2\pi)\right)^2\,\left(E_{\eta^{(\epsilon,\delta)};\frac{\pi}{2}}\ast \Da\right)(\alpha)&&\\
+\int_{q-\delta}^{2\pi}\text{d}\alpha\,\left(\eta^{(\epsilon,\delta)}(\alpha-q)\right)^2\,\left(E_{\eta^{(\epsilon,\delta)};\frac{\pi}{2}}\ast \Da\right)(\alpha)&&2\pi-\delta\leq q<2\pi
\end{array}
\right. 
\end{equation}
The behaviour of $\check{q}$ is depicted in the figure \ref{fig:lowersymbol} in order to be compared with $E_{\eta^{(\epsilon,\delta)};\frac{\pi}{2}}\ast\Da$.
\begin{figure}
    \centering
    \begin{subfigure}[b]{0.6\textwidth}
        \includegraphics[width=\textwidth]{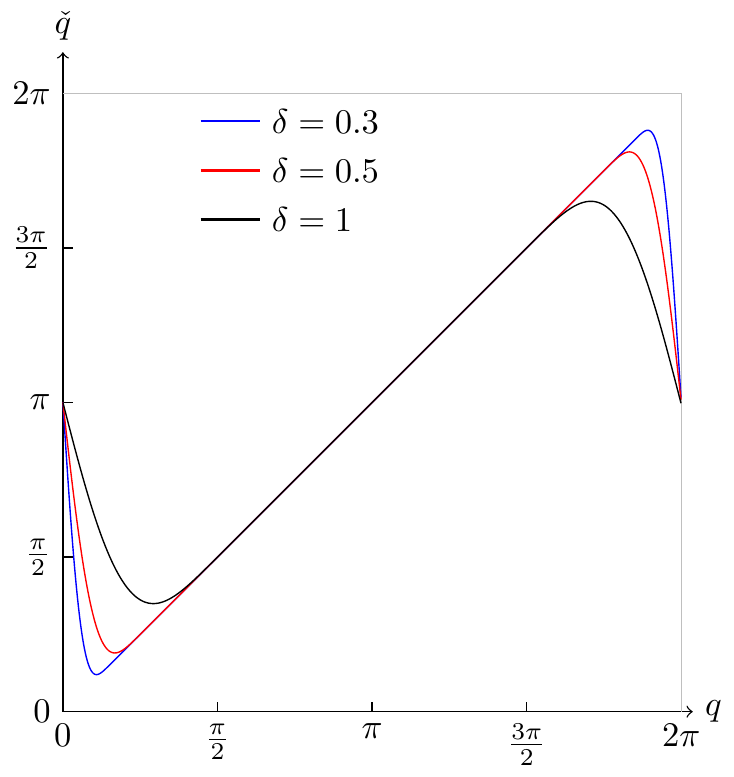}
        \caption{$\epsilon=1$}
        \label{fig:lowerdelta}
    \end{subfigure}
    \begin{subfigure}[b]{0.6\textwidth}
        \includegraphics[width=\textwidth]{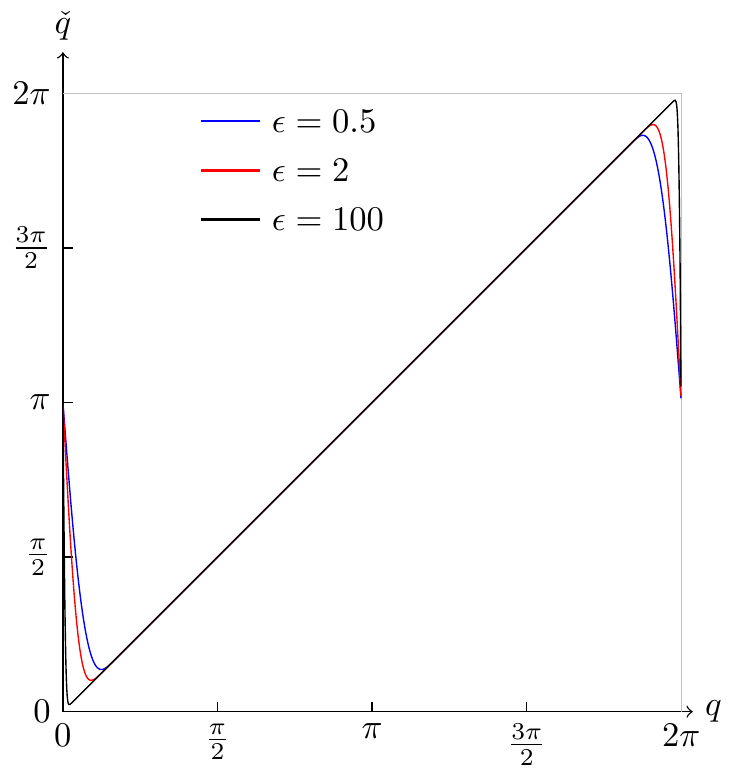}
        \caption{$\delta=0.3$}
        \label{fig:lowereps}
    \end{subfigure}
   \caption{For the lower symbol $\check{q}(q)$ of the angle operator $A_{\Da}$, we check that    $\check{q}(q)\to\Da(q)$ as $\delta \to 0$ or $\epsilon\to\infty$.}\label{fig:lowersymbol}
\end{figure}
%

%---------------------------------------------------------------------------------------
%\section{Next step: extract from physics some nice illustration of this formalism}

\section{Angle-Momentum: commutation and inequality}
\label{heisenberg}
\subsection{Commutation relation}
For $\lambda=0$ and $\psi(\alpha)\in L^{2}(\mathbb{S}^1,\text{d}\alpha)$, using \eqref{eq:quantp} we find the following (non-canonical) commutation rule between the angle operator and the momentum operator,
\begin{equation}
\label{commpq}
\left([A_p, A_{\Da}]\psi\right)(\alpha)= -\ii \mathtt{c}\left(\frac{\ud}{\ud\alpha}\left(E_{\eta;\gamma}\ast\Da\right)(\alpha) \right) \psi(\alpha)\,, \quad  \mathtt{c}:=\frac{c_{2}(\eta,\gamma)}{ \kappa c_{1}(\eta,\gamma)}  .
\end{equation}
Considering Eq. \eqref{anglefunction} and applying Eq. \eqref{derFalp} we arrive at the following result,
\begin{equation}
\label{commpq2}
%\begin{split}
\left([A_p, A_{\Da}]\psi\right)(\alpha)= \left\lbrace\begin{array}{cc}-\ii \mathtt{c}\left(1-2\pi E_{\eta;\gamma}(\alpha) \right)\psi(\alpha)& \mbox{for}\quad \alpha\neq 2k\pi\, ,\\
  0 & \mbox{for}\quad \alpha =  2k\pi \end{array}\right.\,. 
%\end{split}
\end{equation}

Note that the commutator function $\left([A_p, A_{\Da}]\psi\right)(\alpha)$ is discontinuous at $\alpha=2k\pi$. In order to give $A_p$ its usual form $-\ii\partial/\partial\alpha$, a suitable choice on $\kappa$ can be made. We recall that it is always possible to choose $\eta$ such that $\mathtt{c}$ is equal to $1$, as it was mentioned above \eqref{kappachoice}. Using appropriate localisation conditions on the fiducial vector $\eta$, it is also possible to recover the classical canonical commutation rule except for its value at the origin $\text{mod}\,2\pi$.

There are two possible choices:
\begin{itemize}
\item First choice, $\kappa=\frac{c_{2}(\eta,\frac{\pi}{2})}{c_{1}(\eta,\frac{\pi}{2})}$.
\item Second choice, $\kappa=1$. In this case one applies a limit condition on the expression \eqref{eq:quantp}, 
\begin{equation}
\label{quotientc1c2}
\lim_{\delta\to 0}\frac{c_{2}(\eta^{(\epsilon,\delta)},\frac{\pi}{2})}{c_{1}(\eta^{(\epsilon,\delta)},\frac{\pi}{2})}=1\, .
\end{equation}
\end{itemize}
With the choice \eqref{eq:familydelta} for the fiducial vector, the expression of the commutator operator becomes
\begin{equation}
\label{amcr2}
\left([A_p, A_{\Da}]\psi\right)(\alpha)= \left\lbrace\begin{array}{cc}-\ii \mathtt{c} \left(1-2\pi E_{\eta^{(\epsilon,\delta)};\frac{\pi}{2}}(\alpha)\right)\psi(\alpha)&\mbox{for}\quad \alpha\neq 2k\pi\, ,\\
0 &\mbox{for}\quad \alpha =  2k\pi \end{array}\ \right. \, .
\end{equation}
Selecting one of the two possible choices on $\kappa$ and applying the limit $\delta\to 0$ on the expression \eqref{amcr2}, gives a commutation rule similar to \eqref{comrel1}, with the appearance of a singularity at the origin $\text{mod}\,2\pi$.

\subsection{Heisenberg inequality}
Let us now consider the Heisenberg inequality concerning   the operators angle and angular momentum,
\begin{equation}
\label{deltadeltaphi}
\Delta A_p \,\Delta A_{\Da} \geqslant\frac{1}{2}\left|\bra{\phi}[A_p, A_{\Da}]\ket{\phi} \right|\, ,
\end{equation}
where $\phi\in L^{2}(\mathbb{S}^1,\text{d}\alpha)$ and $\left(\Delta A\right)^2  = \lg\phi|A^2|\phi\rg - \lg\phi|A|\phi\rg^2$. 
\subsubsection{With coherent states}
As discussed in the introduction, one of the main issues regarding the definition of an acceptable angle operator concerns the quantum angular dispersion versus the quantum angular momentum one.  
The Heisenberg inequality or uncertainty relation for the operators $A_p$ and $A_{\Da}$, when  computed with the coherent states $\eta_{p,q}$,  is given by
\begin{equation}
\label{deltadelta}
\Delta A_p \,\Delta A_{\Da} \geqslant\frac{1}{2}\left|\bra{\eta_{p,q}}[A_p, A_{\Da}]\ket{\eta_{p,q}} \right|\, .
\end{equation}
Before calculating directly the product of dispersions on the left-hand side of \eqref{deltadelta}, let us study in more details the right-hand-side of this inequality as a function of phase space variables and underlying constants.  
We {recall that the factor $\mathtt{c}=\frac{c_{2}(\eta,\gamma)}{\kappa c_{1}(\eta,\gamma)}$ can be put equal to $1$ following the remarks after \eqref{kappachoice}. Thus, we just have to study the relative smallness of  the mean value of the multiplication operator given by \eqref{amcr2}.
With the  choice \eqref{eq:familydelta}, for $\widetilde{\eta}^{(\epsilon,\delta)}(\alpha)=\eta^{(\epsilon,\delta)}(-\alpha)$ the r.h.s. of \eqref{deltadelta} reads 
\begin{align}
\label{braketF}
\nonumber
\frac{1}{2}\left|\bra{\eta_{p,q}} [A_p, A_{\Da}]\ket{\eta_{p,q}}\right| 
& =\frac{1}{2}\mathtt{c}\left|1-2\pi\,
\left( \left(\widetilde{\eta}^{(\epsilon,\delta)}\right)^2\ast E_{\eta^{(\epsilon,\delta)};\frac{\pi}{2}}\right) 
(q)
\right|\, .
\end{align}
Let us now compute the l.h.s. of \eqref{deltadelta} by using
\begin{equation}
\label{dispersionsquare}
\left[\Delta A_f(p,q)\right]^{2}=
\left<\eta_{p,q}^{(\epsilon,\delta)}\right|A_{f}^2\left|\eta_{p,q}^{(\epsilon,\delta)}\right>-\left(\left<\eta_{p,q}^{(\epsilon,\delta)}\right|A_{f}\left|\eta_{p,q}^{(\epsilon,\delta)}\right> \right)^{2}\, .
\end{equation}
With the particular case considered  in Section \ref{sec:spectral}, Eq. \eqref{eq:spct}, the value of $\left(\Delta A_{\Da} \right)^{2}$ is given by
\begin{equation}
\label{deviation1}
\begin{aligned}
\left[\Delta A_{\Da}(q) \right]^{2}=\left[\left(\widetilde{\eta}^{(\epsilon,\delta)}\right)^2\ast\left( E_{\eta^{(\epsilon,\delta)};\frac{\pi}{2}}\ast\Da\right)^2\right](q)
-
\left(\left[\left(\widetilde{\eta}^{(\epsilon,\delta)}\right)^2\ast\left( E_{\eta^{(\epsilon,\delta)};\frac{\pi}{2}}\ast\Da\right)\right](q) \right)^{2}\, ,
\end{aligned}
\end{equation}
\begin{figure}[H]
    \centering
	        \includegraphics[width=0.7\textwidth]{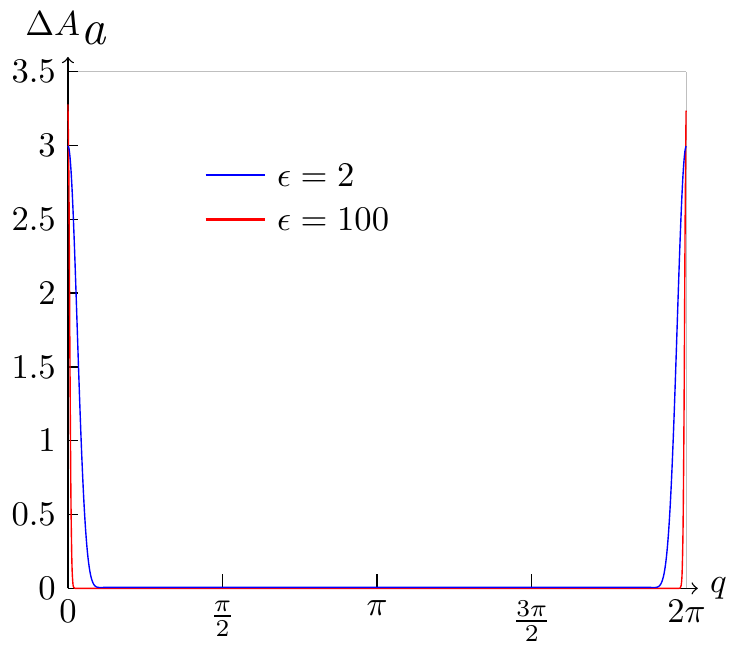}
	        \caption{Behavior of the dispersion $\Delta A_{\Da}$ for $\delta=0.3$ and $\epsilon\geq2$, when they are defined with respect to the coherent state $\ket{\eta_{p,q}^{(\epsilon,\delta)}}$.}
	        \label{fig:qvar}
\end{figure}
In  figure \ref{fig:qvar}, $\Delta A_{\Da}$ is almost constant with a very small value (good localisation). The exception corresponds to the region near $q=0$ and $q=2\pi$, where we see that it is not possible to have a good localisation since the value of $\Delta A_{\Da}$ is too large.

The value of $\left[\Delta A_p(p)\right]^{2}$ is given by
\begin{equation}
\label{deviationp}
\begin{aligned}
\left[\Delta A_p(p)\right]^{2}=
\mathtt{c}^2
\frac{\epsilon}{\delta^{2}}\frac{1}{e_{2\epsilon}}\int_{-1}^{1}\ud x \,
\omega_{2\epsilon}(x)
\left(\frac{2\left(1+3x^2\right)}{(1-x^2)^{3}}-\epsilon\frac{4x^2}{(1-x^2)^{4}}\right)\\
+
\left(c_{-2}\left(\eta^{(\epsilon,\delta)},\frac{\pi}{2}\right)-c_{-1}\left(\eta^{(\epsilon,\delta)},\frac{\pi}{2}\right)^{2} \right) (\kappa \mathtt{c})^2\,p^2\, .
\end{aligned}
\end{equation}
\begin{figure}[H]
    \centering
	        \includegraphics[width=0.7\textwidth]{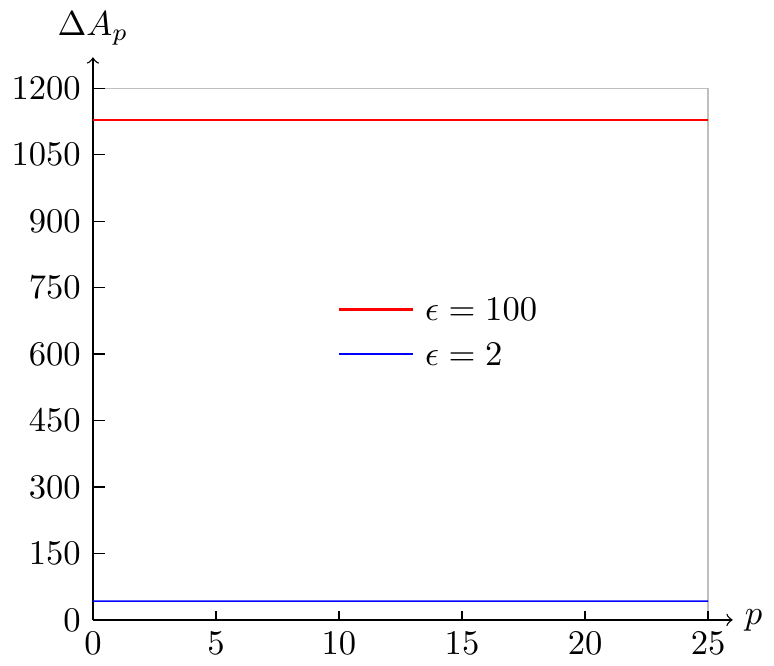}
	        \caption{Behavior of the dispersion $\Delta A_{p}$ for $\delta=0.3$ and $\epsilon\geq2$, when they are defined with respect to the coherent state $\ket{\eta_{p,q}^{(\epsilon,\delta)}}$.}
	        \label{fig:pvar}
\end{figure}
As is seen in Figure \ref{fig:pvar}, with a suitable choice of $\epsilon$ and $\delta$, the dependence on $p$ in the expression \eqref{deviationp} can be neglected, therefore $\Delta A_p$ does not depend on $p$ or $q$. 

In figure \ref{fig:uncertaintyR} we show that the 
uncertainty product $\Delta A_{\Da}\Delta A_p$ for $\delta=0.3$ gets close to the minimum uncertainty with $\epsilon=2$ and $\epsilon=100$. For large values of $\epsilon$ the state $\ket{\eta_{p,q}^{(\epsilon,\delta)}}$ saturates the uncertainty relation \eqref{deltadelta}. 
\begin{figure}[H]	
    \centering
    \begin{subfigure}[b]{0.7\textwidth}
        \includegraphics[width=\textwidth]{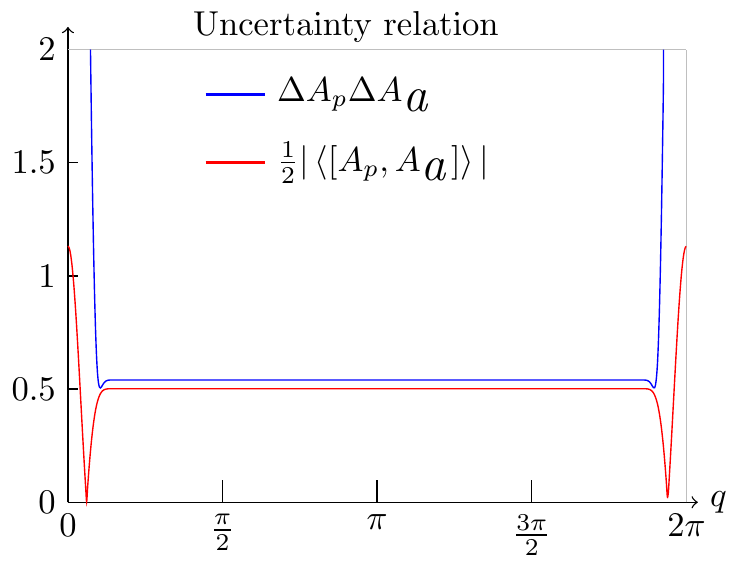}
        \caption{$\epsilon=2$}
        \label{fig:unce2}
    \end{subfigure}
    \begin{subfigure}[b]{0.7\textwidth}
        \includegraphics[width=\textwidth]{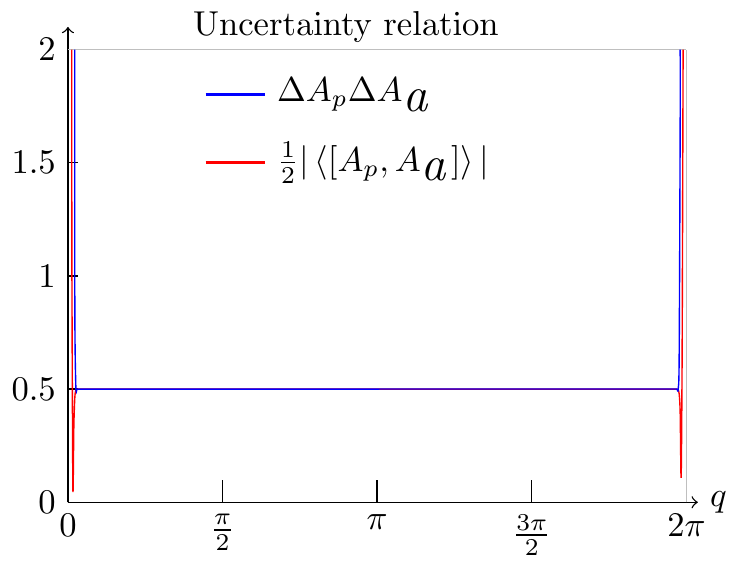}
        \caption{$\epsilon=100$}
        \label{fig:unce100}
    \end{subfigure}
    \caption{Behavior of the uncertainty relation for $\delta=0.3$ and $\epsilon\geq2$, respect to the state $\ket{\eta_{p,q}^{(\epsilon,\delta)}}$.}\label{fig:uncertaintyR}
\end{figure}
\subsubsection{With Fourier exponentials as eigenfunctions of $A_{p}$}
It is interesting to compare the above inequalities computed with coherent states with those calculated from the eigenstates $\varphi_{m}(\alpha)=\frac{e^{\text{i}m\alpha}}{\sqrt{2\pi}}$ of $A_{p}$, $m\in\mathbb{Z}$,
\begin{equation}
\left(A_{p}^n\varphi_{m}\right)(\alpha)= \mathtt{c}^n m^n \varphi(\alpha)\, ,
\end{equation}
and obviously $\left(\Delta A_{p} \right)^2 =0$. 
 The action of $A_{\Da}$ on $\varphi_{m}$ is
\begin{equation}
\left(A_{\Da}^n\varphi_{m}\right)(\alpha)= \left[ \left(E_{\eta;\gamma}\ast\Da \right) (\alpha)\right] ^n \varphi(\alpha)\, .
\end{equation}
The expectation value $\left<\varphi_{m}\right|A_{\Da}^n\left|\varphi_{m}\right>$ will be
\begin{equation}
\left<\varphi_{m}\right|A_{\Da}^n\left|\varphi_{m}\right>=
\int_{0}^{2\pi}\text{d}\alpha\,\overline{\varphi_{m}(\alpha)}\left(A_{\Da}^n\varphi_{m}\right)(\alpha)
=\frac{1}{2\pi}\int_{0}^{2\pi}\text{d}\alpha\,\left[ \left(E_{\eta;\gamma}\ast\Da \right) (\alpha)\right] ^n\, .
\end{equation}
Considering $\delta=0.3$ one can calculate $\Delta A_{\Da}$ for different values of $\epsilon$:
\begin{center}
 \begin{tabular}{||c c c c||} 
 \hline
 $\epsilon$ & $\left<A_{\Da}\right>$ & $\left<A_{\Da}^2\right>$ & $\Delta A_{\Da}$ \\ [0.5ex] 
 \hline\hline
 $1$ & $3.1404$ & $12.8007$ & $1.7143$ \\ 
 \hline
 $2$ & $3.1404$ & $12.8662$ & $1.7333$ \\
 \hline
 $10$ & $3.1403$ & $13.0016$ & $1.7720$ \\
 \hline
 $100$ & $3.1402$ & $13.1040$ & $1.8009$ \\ [1ex] 
 \hline
\end{tabular}
\end{center}
The above values are close to the dispersion for \eqref{angleop1}, where $\Delta\hat{\theta}=\frac{\pi}{\sqrt{3}}=1.8138$. Of course there is no contradiction with the inequality \eqref{deltadeltaphi}, since the average value of the commutator in the normalised Fourier exponentials is also vanishing. 

%\newpage
\section{Fourier analysis and other probabilistic aspects}
In this section, we  compare our results with the ones in \cite{argaho12}, where the construction of the angle operator was achieved through a quantisation based on different coherent states for the motion on the circle.
With the notations of  Appendix \ref{CirCSqangle} where we give a review of the work \cite{argaho12}, these normalised states are defined by the expansion (see \eqref{discrprobC}) 
\begin{equation}
\label{ccs1}
| p, q \rangle =  \frac{1}{\sqrt{{\mathcal N}^{\sigma} (p)}} \sum_{n \in \Z} \sqrt{w^{\sigma}_n(p)}  \,\E^{- \ii n q} | e_n\rangle\, ,\quad p \in \R\, , 0\leq q < 2\pi\, 
\end{equation}
where the $|e_n\rg$'s form an orthonormal basis of a Hilbert space,   $\mathcal{H}$, and $0< \mathcal{N}^{\sigma}(p) \deq \sum_{n\in \Z}  w^{\sigma}_n (p)$ is the normalisation factor. The construction of states \eqref{ccs1} rests upon the  probability distribution $w^{\sigma} (p)$ on the range of the  variable $p$. It is a non-negative, \underline{even}, well-localized and normalized integrable function which is {subject to the conditions listed in \ref{condcs}. 

States \eqref{ccs1} resolve the identity with respect to the measure $\mathcal{N}^{\sigma} (p)\,  \frac{\ud p\, \ud q}{2\pi}$.  By virtue of the CS quantisation scheme,  the quantum operator (acting on ${\mathcal H}$) associated with the functions $f(p,q)$ on the cylinder is obtained through
\begin{equation}
A_f := \int_{ { \R}\times [0,2\pi] }f(p,q) | p, q\rangle \langle p, q| \,\mathcal{N}^{\sigma} (p)\,  \frac{\ud p\, \ud q}{2\pi} = \sum_{n,n'} \left(A_f\right)_{nn'} \, |e_n\rg \lg e_{n'}|\, .
\label{ccs2}
\end{equation}
In particular, we obtain  the  angle operator corresponding to the $2\pi$-periodic angle function $\Da (q)$ previously defined as the  periodic extension of $\Da (q) = q$  for $0\leq  q< 2\pi$:
\begin{equation}
\label{ccs3}
A_{\Da } = \pi I + \ii\, \sum_{n\neq n^{\prime}}\frac{w^{\sigma}_{n,n^{\prime}}}{n-n^{\prime}}\,| e_n\rangle \langle e_{n^{\prime}} |\, .
\end{equation}
This operator is bounded self-adjoint. For $f(p) = p$, the assumptions on $w^{\sigma}$ give
\begin{equation}  
A_p  = \int_{\mathbb{S}^1\times\R} \frac{\ud p\, \ud\alpha}{2\pi}  \mathcal{N}^{\sigma}(p)\, p\, | p,\alpha \rangle \langle p, \alpha |  = \sum_{n \in \Z}
n\, | e_n\rangle \langle e_n| = N\, .
%\label{Jsym}
\end{equation}
This is nothing  but the number or angular momentum operator (in units of $\hbar$), which reads  $A_p = -\ii \partial/\partial \theta$ in angular position representation, i.e. when $\mathcal{H}$ is chosen as $L^2(\mathbb{S}^1,\ud\alpha)$ with orthonormal basis $|e_n\rg \equiv \E^{\ii n\alpha}/\sqrt{2\pi}$ (Fourier series).

Let us compare states \eqref{ccs1} with the CS introduced  in  the present paper and given by \eqref{qpalfa}. Their  Fourier expansion reads
\label{fourier}
\begin{equation}
\label{fourierdeceta}
\eta_{p,q}(q)=\sum_{n\in\mathbb{Z}}c_{n}\left(\eta_{p,q}\right)e^{\textrm{i}n\alpha}\, ,
\end{equation}
with Fourier coefficients given by
\begin{equation}
\label{fouriercoeffeta}
c_{n}\left(\eta_{p,q}\right)=\frac{1}{2\pi}\int_{\mathbb{S}^{1}}\textrm{d}\alpha\, e^{-\textrm{i}n\alpha}\eta_{p,q}(\alpha)=\frac{1}{2\pi}\int_{\mathbb{S}^{1}}\textrm{d}\alpha\, e^{-\textrm{i}n\alpha}e^{\textrm{i}\left[\mathcal{R}(q-\alpha)(\kappa p+\lambda)\right]_{1}}\eta(\alpha-q)\, .
\end{equation}
The change of variable $\alpha\rightarrow \alpha-q$ in \eqref{fouriercoeffeta} gives
\begin{equation}
\label{fouriercoeffeta2}
c_{n}\left(\eta_{p,q}\right)=e^{\textrm{-i}nq}\frac{1}{2\pi}\int_{\mathbb{S}^{1}}\textrm{d}\alpha\, e^{-\textrm{i}n\alpha+\textrm{i}\kappa p\cos(\alpha-\gamma)+\textrm{i}\lambda\cos(\alpha-\zeta)}\eta(\alpha)\, .
\end{equation}
Comparing the Fourier coefficients \eqref{ccs1} with the ones given in \eqref{fouriercoeffeta2}  yields the relation
\begin{equation}
\label{integrablefunction}
\sqrt{w_{n}^{\sigma}(p)}=
\frac{\sqrt{\mathcal{N}^{\sigma}(p)}}{2\pi}
\int_{\mathbb{S}^{1}}\textrm{d}\alpha\, e^{-\textrm{i}n\alpha+\textrm{i}\kappa p\cos(\alpha-\gamma)+\textrm{i}\lambda\cos(\alpha-\zeta)}\eta(\alpha)\, .
\end{equation}
Besides the positiveness condition imposed to the r.h.s integral above,  we immediately notice that  \eqref{integrablefunction} fails to fullfill  the condition
\eqref{w0wn}, i.e. $\sqrt{w_{n}^{\sigma}(p)} = \sqrt{w_{0}^{\sigma}(p-n)}$. Hence,  it is not possible to make a direct connection between \cite{argaho12} and our present work.

%\label{fourier}

\section{Conclusions}

In this work we address the open problem concerning the quantisation of the angle operator and its related localisation following the method of covariant integral quantisation.  
The cylinder $\R\times\mathbb{S}^{1}$ depicts the classical phase space of the motion of a particle on a circle, which is mathematically realized as the cotangent bundle $\text{E}(2)/H_0$, where $H_0$ is the stabilizer under the coadjoint action of $\text{E}(2)$. 
The coherent states for $\text{E}(2)$  are constructed from the induced representations of the semi-direct product structure $\text{E}(2)=\R^{2}\rtimes\text{SO}(2)$. For various functions on phase space, the corresponding operators are provided. 
In the particular case of periodic functions $f(q)$ of the coordinate $q$, the operators $A_f$  are multiplication operators whose spectra are given by periodic functions. 

The angle function $\Da(\alpha)$, defined by $\Da(\alpha) = \alpha $ for $\alpha\in[0,2\pi)$, is  mapped to  a self-adjoint multiplication  angle operator $A_{\Da}$ with continuous spectrum. For a particular family of coherent states,  it is shown that the spectrum is $[\pi -m(\epsilon,\delta), \pi + m(\epsilon,\delta)]$, where $m(\epsilon,\delta)\to \pi$ as $\delta \to 0$ or $\epsilon\to\infty$. In other words, we are restricted to the motion on $[\pi -m(\epsilon,\delta), \pi + m(\epsilon,\delta)]$, the whole circle is recovered only when $\delta \to 0$ or $\epsilon\to\infty$. Therefore systems like the classical pendulum or the torsion spring (where the angular motion is restricted) can be quantised without major issues.
Is also shown that the semiclassical portrait $\check{q}$ of $A_{\Da}$ can be made arbitrarily close to the values of the angle function $\Da(\alpha)$.

We found a (non-canonical) commutation rule between the angle operator and the momentum operator, as well as an expression for the uncertainty relation between them. The uncertainty relation with eigenstates of the momentum gives similar results to what we expect at working with \eqref{angleop1}.
\label{conclu}

\appendix
%\section{Some constants and other quantities used}
%\label{app:list}

\section{Quantum angle for cylindric phase space}
\label{CirCSqangle}
 
  In this appendix, we give a summary of the work \cite{argaho12} where other  coherent states for the motion on the circle and their associated integral quantisation quantisation were presented.
 
 \subsection{Other coherent states}
 
We start with  the cylindric phase space $\R \times  {[0,2\pi]} = \{  (p,q), \, | \, \, p \in \R\, , \, 0 \leq q < 2\pi   \}$, equipped with the measure $ \frac{1}{2\pi}  \, \ud p\, \ud q $.  We introduce a probability distribution on the range of the  variable $p$. It is a non-negative, \underline{even}, well localized and normalized integrable  function 
\begin{equation}
%\label{gprobdist}
\R\ni p \mapsto w^{\sigma} (p)\,  ,\quad w^{\sigma} (p)= w^{\sigma} (-p)\,, \quad \int_{-\infty}^{+\infty} \ud p\, 
w^{\sigma} (p)=1\, , 
\end{equation}
where $\sigma >0$ is a  width parameter. 
This function must obey the following conditions:
\begin{proper}
\label{condcs}
\begin{itemize}
 \item[(i)] $0< \mathcal{N}^{\sigma}(p) \deq \sum_{n\in \Z}  w^{\sigma}_n (p) < \infty$ for all $p\in \R$, where
 \begin{equation}
\label{w0wn}
w^{\sigma}_n (p)\deq w^{\sigma} (p-n)\,. 
\end{equation} 
 \item[(ii)] the Poisson summation formula is applicable to $\mathcal{N}^{\sigma}$:
 \begin{equation*}
 %\label{poissvp}
 \mathcal{N}^{\sigma}(p) = \sum_{n\in \Z}  w^{\sigma}_n (p) = 
 \sqrt{2\pi} \sum_{n\in \Z}  \E^{-2\pi \ii np}\widehat{w^{\sigma}}(2\pi n)\, ,  
 \end{equation*}
 where $\widehat{w^{\sigma}}$ is the Fourier transform of $w^{\sigma}$,
 \item[(iii)]  its limit  at $\sigma \to 0$, in a distributional  sense, is the Dirac distribution:
 \begin{equation*}
 %\label{varpidi}
 w^{\sigma} (p) \underset{\sigma \to 0}{\to} \delta(p)\,,
 \end{equation*} 
 \item[(iv)]  the limit  at $\sigma \to \infty$ of its Fourier transform is proportional to the characteristic function of the singleton $\{0\}$:
   \begin{equation*}
 %\label{varpidir}
 \widehat{w^{\sigma}} (k) \underset{\sigma \to \infty}{\to} \frac{1}{\sqrt{2\pi}}\, \delta_{k0}\,,
 \end{equation*} 
 \item[(v)] considering the \emph{overlap matrix} of the two distributions $p \mapsto w^{\sigma}_n(p)$, $p \mapsto w^{\sigma}_{n'}(p)$ with matrix elements,
  \begin{equation*}
 %\label{varpcor}
 w^{\sigma}_{n,n'} = \int_{-\infty}^{+\infty} \ud p\, \sqrt{w^{\sigma}_n(p)\, w^{\sigma}_{n'}(p)} \leq 1\, , 
 \end{equation*}
 we impose the two conditions 
% \begin{enumerate}
%   \item[(a)] 
  \begin{equation*}
 \label{cond1}
  w^{\sigma}_{n,n'} \to 0 \quad \mbox{as} \quad  n-n' \to \infty  \quad \mbox{at fixed}\ \sigma\, , \tag{a}
  \end{equation*} 
    % \item[(b)]
 \begin{equation*}    
  \label{cond2}  
  \exists\, n_M\geq 1 \quad \mbox{such that} \quad   w^{\sigma}_{n,n'} \underset{\sigma \to \infty}{\to} 1 \quad   % \forall\, n, \, n'\, \ 
  \mbox{provided} \ \vert n-n'\vert \leq n_M \, .\tag{b}
 \end{equation*} 
% \end{enumerate}
\end{itemize} 
\end{proper}
Properties (ii) and (iv) entail that $ \mathcal{N}^{\sigma}(p)  \underset{\sigma \to \infty}{\to} 1$. Also note the properties of the overlap matrix elements $w^{\sigma}_{n,n'}$ due to the properties of $w^{\sigma}$:
\begin{equation*}
%\label{propcorr}
w^{\sigma}_{n,n'}= w^{\sigma}_{n',n}=w^{\sigma}_{0,n'-n}= w^{\sigma}_{-n,-n'}\,, \quad w^{\sigma}_{n,n} = 1\, \quad \forall\, n, n' \in \Z\, .
\end{equation*}
The most immediate  choice for $w^{\sigma} (p)$  is  Gaussian, i.e. $w^{\sigma} (p)= \sqrt{\frac{1}{2\pi\sigma^2}}\,\E^{-\frac{1}{2\sigma^2} p^2 }$ (for which the $n_M$ in \eqref{cond2} is $\infty$) 
 Let us now introduce the weighted Fourier exponentials:
\begin{equation*}
%\label{circphin1}
\phi_n (p,\alpha) = \sqrt{w^{\sigma}_n(p)}  \,\E^{  \ii n\alpha}\, , \quad n\in \Z\,.
\end{equation*}
These functions form the countable  orthonormal system in $L^2(\mathbb{S}^1\times \R,\ud p\,\ud q/2\pi)$ needed to construct coherent states in agreement with a general  procedure explained, for instance,  in \cite{gazbook09}. Let $\mathcal{H}$ be a Hilbert space with orthonormal basis $\{| e_n\rangle\, | \, n\in \Z\}$, e.g. $ \mathcal{H} = L^2(\mathbb{S}^1, \ud \alpha\}$ with $e_n(\alpha) = \frac{1}{\sqrt{2\pi}}e^{\ii n \alpha}$. The  correspondent family of coherent states on the circle reads as:
\begin{equation}
%\label{ccs}
| p, q \rangle =  \frac{1}{\sqrt{{\mathcal N}^{\sigma} (p)}} \sum_{n \in \Z} \sqrt{w^{\sigma}_n(p)}  \,\E^{- \ii n q} | e_n\rangle\, .
\end{equation}
 These states  are normalized and resolve the unity. They overlap as:
\begin{equation*}
%\label{overlap}
\lg p,q|p^{\prime},q^{\prime}\rg =  \frac{1}{\sqrt{{\mathcal N}^{\sigma} (p)\, {\mathcal N}^{\sigma} (p\textit{})}}\sum _{n \in \Z} \sqrt{w^{\sigma}_n(p)\, w^{\sigma}_n(p^{\prime})}  \,\E^{- \ii n(q-q^{\prime})}\, . 
\end{equation*}
 The function $w^{\sigma} (p)$ gives rise to a double probabilistic interpretation \cite{gazbook09}:
\begin{itemize}
 \item For all $p$ viewed as a shape parameter, there is the discrete distribution, 
 \begin{equation}
 \label{discrprobC}
 \Z \ni n \mapsto \vert \lg e_n|p,\alpha\rg\vert^2=  \frac{ w^{\sigma}_n (p)}{{\mathcal{N}}^{\sigma} (p)} \, .
 \end{equation}
 This probability, of genuine quantum nature,   concerns experiments performed on the system described by the Hilbert space $\mathcal{H}$ within some experimental protocol, %say $\mathcal{E}$, 
 in order to measure the  spectral values of a  self-adjoint operator  acting in $\mathcal{H}$ and having the discrete spectral resolution $ \sum_n a_n |e_n\rg\lg e_n|$. For $a_n=n$ this operator is  the number or quantum angular momentum operator. 
 
 \item For each $n$, there is the continuous distribution on the cylinder $\R \times \mathbb{S}^1$ (reps. on $\R$) equipped with its  measure $\ud p\, \ud q/2\pi$ (resp. $\ud p$), 
 \begin{equation}
 \label{contprobC}
 (p, q) \mapsto \vert \phi_n (p, q) \vert^2 = w^{\sigma}_n (p)\quad (\mbox{resp.}\ \quad \R \ni p \mapsto w^{\sigma}_n (p))\,.
 \end{equation}
 This probability, of classical nature and uniform on the circle,  determines the CS quantisation of functions of $p$.
\end{itemize}
\subsection{CS quantisation}  
By virtue of the CS quantisation scheme,  the quantum operator (acting on ${\mathcal
H}$) associated with functions $f(p,q)$ on the cylinder is obtained through
\begin{equation}
A_f := \int_{ { \R}\times [0,2\pi] }f(p,q) | p, q\rangle \langle p, q| \,\mathcal{N}^{\sigma} (p)\,  \frac{\ud p\, \ud q}{2\pi} = \sum_{n,n'} \left(A_f\right)_{nn'} \, |e_n\rg \lg e_{n'}|\, ,
\label{oper15}
\end{equation}
where
\begin{equation}
\label{matelAf15}
\left(A_f\right)_{nn'} = \int_{-\infty}^{+\infty}\ud p\, \sqrt{w^{\sigma}_n(p)\, w^{\sigma}_{n'}(p)}\,\frac{1}{2 \pi}\int_0^{2 \pi}\ud q\, \E^{-\ii (n-n')q}\, f(p,q)\, .
\end{equation}
The lower symbol of $f$ is  given by:
\begin{equation}
\label{lowsymbvp}
\check{f} (p,q) = \lg p,q | A_f | p,q \rg 
= \int_{-\infty}^{+\infty}\ud p' \int_{0}^{2\pi}
\frac{\ud q^{\prime}}{2\pi}\,\mathcal{N}^{\sigma}(p^{\prime}) \, f(p^{\prime},q^{\prime})\, \vert\lg p,q| p^{\prime},q^{\prime}\rg\vert^2 \, .
\end{equation}
If $f$ is depends on $p$  only, $f(p,q) \equiv v(p)$, then $A_f$ is diagonal with matrix elements that are $w^{\sigma}$ transforms of $v(p)$:
\begin{equation*}
%\label{diagmatelAf15}
\left(A_{v(p)}\right)_{nn^{\prime}} = \delta_{nn^{\prime}}\int_{-\infty}^{+\infty}\ud p\, w^{\sigma}_n(p)\, v(p)= \delta_{nn^{\prime}} \lg v\rg_{w^{\sigma}_n}\, ,
\end{equation*}
where $ \lg \cdot \rg_{w^{\sigma}_n}$ designates the mean value w.r.t. the distribution $p\mapsto w^{\sigma}_n(p)$. 
For the most basic case, $v(p) = p$,    our assumptions  on $w^{\sigma}$ give
\begin{equation}  A_p  = \int_{\mathbb{S}^1\times\R} \frac{\ud p\, \ud\alpha}{2\pi}  \mathcal{N}^{\sigma}(p)\, p\, | p,\alpha \rangle \langle p, \alpha |  = \sum_{n \in \Z}
n\, | e_n\rangle \langle e_n| = N\, .
%\label{Jsym}
\end{equation}
This is nothing  but the number or angular momentum operator (in unit $\hbar = 1$), which reads  $A_p = -\ii \partial/\partial \alpha$ in angular position representation, i.e. when $\mathcal{H}$ is chosen as $L^2(\mathbb{S}^1,\ud\alpha)$. 

Let us define the unitary representation $\theta \mapsto U_{\mathbb{S}^1}(\theta)$ of  $\mathbb{S}^1$ on  $\mathcal{H}$ as the diagonal operator  $U_{\mathbb{S}^1}(\theta)|e_n\rg = \E^{\ii n  \theta}|e_n\rg$, i.e. $U_{\mathbb{S}^1}(\theta) = \E^{\ii \theta N}$. We easily infer from  
the  straightforward covariance property of the coherent states :
\begin{equation*}
%\label{covCS}
U_{\mathbb{S}^1}(\theta)|p,q\rg = |p, q - \theta\rg\,  ,
\end{equation*}
the  rotational covariance of $A_f$ itself,
\begin{equation*}
%\label{rotcovAf}
U_{\mathbb{S}^1}(\theta)A_f U_{\mathbb{S}^1 }(-\theta)= A_{T^{-1}(\theta)f} \, , 
\end{equation*}
where $T^{-1}(\theta)f(\alpha)\okr f(\alpha + \theta)$.
 
If $f$ depends on $q$ only,   $f(p,q)= u(q)$, we have
\begin{align}  A_{u(q)} = & \int_{\R\times [0,2\pi]} \frac{\ud p\, \ud q}{2\pi}\mathcal{N}^{\sigma}(p) v(q) \,  | p,q \rangle \langle p, q |  \\
&= \sum_{n,n' \in \Z}
w^{\sigma}_{n,n'} \,c_{n-n'}(v)| e_n\rangle \langle e_{n'} |\,  ,
\label{f(beta)15}
\end{align}
where 
$c_{n}(v)$ is the $n$th Fourier coefficient of $v$. 
In particular, we have  the  angle operator corresponding to the $2\pi$-periodic angle function $\Da (q)$ previously defined as the  periodic extension of $\Da (q) = q$  for $0\leq q < 2\pi$
\begin{equation}
\label{opangle15}
A_{\Da } = \pi I + \ii\, \sum_{n\neq n'}\frac{w^{\sigma}_{n,n'}}{n-n'}\,| e_n\rangle \langle e_{n'} |\, ,
\end{equation}
This operator is bounded self-adjoint. Its covariance property is
\begin{equation}
\label{covqancirc}
U_{\mathbb{S}^1}(\theta)A_{\Da } U_{\mathbb{S}^1}(-\theta) = A_{\Da } + (\theta \, \mathrm{mod}(2\pi))I\, .
\end{equation}
Note the operator corresponding to the elementary Fourier exponential,
\begin{equation}
\label{opfourier15}
A_{\E^{\pm \ii q}} = \, w^{\sigma}_{1,0}\sum_{n}
| e_{n \pm 1}\rangle \langle e_n |\, , \quad A_{\E^{\pm \ii q}}^{\dagger}= A_{\E^{\mp \ii q}}\, . 
\end{equation}
We remark that $A_{\E^{\pm \ii q}}\, A_{\E^{\pm \ii q}}^{\dagger}=  A_{\E^{\pm \ii q}}^{\dagger}\,A_{\E^{\pm \ii q}}= (w^{\sigma}_{1,0})^2 1_d$. Therefore this operator fails to be unitary. It is ``asymptotically'' unitary at large $\sigma$ since the factor $(w^{\sigma}_{1,0})^2$ can  be made  arbitrarily close to $1$ at large $\sigma$ as a consequence of Requirement (\ref{cond2}). In the Fourier series realization of ${\mathcal H}$, for which the kets $ | e_n \rangle$
are the Fourier exponentials $\E^{\ii \,n\alpha}/\sqrt{2\pi}$,  the operators $A_{\E^{\pm \ii q}} $ are multiplication operators by  $\E^{\pm \ii \alpha}$ up to the factor $w^{\sigma}_{1,0}$. 
Finally,  the  commutator of angular momentum and angle operators is given by the expansion
\begin{equation}
\label{ccrcir15}
\lbrack A_p, A_{\Da } \rbrack = \ii \sum_{n \neq n'}
w^{\sigma}_{n,n'}\, | e_n\rangle \langle e_{n' } |\, .
\end{equation} 
One observes that the overlap matrix completely encodes  this basic commutator.
 Because of  the required properties of the distribution $w^{\sigma}$ the departure of the r.h.s. of \eqref{ccrcir15} from the canonical r.h.s. $-\ii I$ can be bypassed by examining the behavior of the lower symbols at large $\sigma$. For an original function depending on $q$ only we have the Fourier series
\begin{equation}
\label{lowsymbfphi}
\check{f}(p_0, q_0)=  \langle p_0, q_0 | A_{f} |  p_0, q_0 \rangle    =  c_0(f) + \sum_{m \neq 0} d_m^{\sigma}(p_0)\,  w^{\sigma}_{0,m}\, c_m(f)\, \E^{ \ii m q_0}\, , 
\end{equation}
with 
\begin{equation}
\label{coeffdfphi}
d_m^{\sigma}(p)= \frac{1}{\mathcal{N}^{\sigma}(p)}\sum_{r=-\infty}^{+\infty} \sqrt{w^{\sigma}_{r}(p)w^{\sigma}_{m+r}(p)} \leq 1\, ,
\end{equation}
the last inequality resulting from Condition (i) and Cauchy-Schwarz inequality. If we further  impose the condition that $d_m^{\sigma}(p) \to 1$ uniformly as $\sigma \to{+} \infty$, then  the lower symbol $\check{u}(p_0, q_0)$ tends to the Fourier series of the original function $u(q)$. A similar result is obtained for the lower symbol of the commutator (\ref{ccrcir15}):
\begin{equation}
\label{lowsymcomangle}
\langle p_0, q_0 | \lbrack A_p, A_{\Da} \rbrack |  p_0, q_0 \rangle    =  
\ii \sum_{m \neq 0} d_m^{\sigma}(p_0)\,  w^{\sigma}_{0,m}\, \E^{ \ii m q_0}\, . 
\end{equation}
 Therefore, with the condition that $d_m^{\sigma}(p) \to 1$ uniformly as $\sigma \to \infty$, we obtain at this limit  the result similar to \eqref{comrel1},
\begin{equation}
\label{lowsymcomlim}
\langle p_0, q_0 | \lbrack A_p, A_{\Da} \rbrack |  p_0, q_0 \rangle \underset{\sigma \to \infty}{\to}   -\ii +  \ii\sum_{m } \delta(q_0 - 2 \pi m)\, . 
\end{equation}
So we asymptotically (almost) recover the classical canonical commutation rule except for the singularity at the origin $\mathrm{mod}\, 2\pi$, a logical consequence of the discontinuities of the saw function $\Da (q)$ at these points.

%---------------------------------------------------------------------------------------
%%%%%%%%%%%%%%%%%%%%%%%%%%%%%%%%%%%%%%%%%%%%%%%%%%%%%%%%%%%%%%%%%%%%%%%%%%%%%
%\bibliographystyle{plain}
%\bibliography{bbl}

\begin{thebibliography}{99}

\bibitem{levyleblond76} J.M. L\'evy-Leblond, Who is afraid of nonhermitian operators? A quantum description of angle and phase,  \textit{Annals of Physics (NY)} \textbf{101} 319-341
(1976).


\bibitem{kowal02} K. Kowalski, K. Podlaski, and J. Rembieli\'nski, 
Quantum mechanics of a free particle on a plane with an extracted point, 
\textit{Phys. Rev. A} \textbf{66} 032118-1-9  (2002). 

\bibitem{dirac27} P.A.M. Dirac, The Quantum Theory of the Emission and Absorption of Radiation, \textit{Proc. R. Soc. London} \textbf{A114} 243-263  (1927).

\bibitem{dirac58} P.A.M. Dirac, \textit{Principles of quantum mechanics}, Oxford, 1958.

\bibitem{carruthers1968} P. Carruthers and M.~M. Nieto, Phase and Angle Variables in Quantum Mechanics, {\it Rev. Mod. Phys.} {\bf 40}  411-440
(1968).

\bibitem{lynch1995} R. Lynch, The quantum phase problem: a critical review, \textit{Phys. Rep.} \textbf{256} 367-436 (1995).

\bibitem{judge1963} D. Judge, On the uncertainty relation for $L_z$ and $\varphi$, \textit{Phys. Lett.} \textbf{5(3)} 189 (1963). 

\bibitem{lewis1963} D. Judge and  J. T.  Lewis, On the commutator $[L_z, \varphi]_{-}$, \textit{Phys. Lett.}  \textbf{5(3)} 190 (1963). 

\bibitem{kraus1965} K. Kraus, Remark on the uncertainty between angle and angular momentum, \textit{Zeitschrift Fur Physik} \textbf{188(4)} 374-377 (1965). 

\bibitem{louisell63} W.H. Louisell, Amplitude and phase uncertainty relations, \textit{Phys. Lett.} \textbf{7} 60-61 (1963).

\bibitem{susskind1964} L. Susskind and J. Glogower, Quantum mechanical phase and time operator, \textit{Physics} \textbf{1}  49 (1964).

\bibitem{lerner1970} E.C. Lerner, H.W. Huang, and G.E. Walters, Some Mathematical Properties
of Oscillator Phase Operators, \textit{J. Math. Phys.} \textbf{11} 1679-1684 (1970).

\bibitem{garrison1970} J. C. Garrison and J. Wong, Canonically Conjugate Pairs, Uncertainty Relations, and Phase Operators, \textit{J. Math. Phys.}, \textbf{11(8)} 2242-2249 (1970). 

\bibitem{alimow1979} A.L. Alimow and E.W. Damaskinski, Self-Adjoint phase operators,
\textit{Theor. Math. Phys.} \textbf{38}
58-70 (1979).

\bibitem{galindo1984} A. Galindo, Phase and number, \textit{Lett. Math. Phys.} \textbf{8(6)} 495-500 (1984). 


\bibitem{ifantis1971} E.K. Ifantis, Abstract Formulation of the Quantum Mechanical Oscillator
Phase Problem, \textit{J. Math. Phys.} \textbf{12} 1021-1026 (1971).

\bibitem{mlak1992}  W. Mlak and M. S{\l}oci\'nski, Quantum phase and circular operators,
\textit{Un. Jagellon. Acta Math.}, Fasc. XXIX, 133-144 (1992).


\bibitem{newton1980} R. G. Newton, Quantum action-angle variables for harmonic oscillators, \textit{Ann. of Phys.} \textbf{124(2)} 327-346 (1980). 

\bibitem{volkin1973} H.C. Volkin, Phase operators and phase relations for photon states,
\textit{J. Math. Phys.} \textbf{14} 1965-1976 (1973).

\bibitem{rocca1973} F. Rocca and M. Sirugue, Phase operator and condensed systems, \textit{Commun. Math. Phys.} \textbf{34(2)} 111-121 (1973). 

\bibitem{royer96}  A. Royer, Phase states and phase operators for the quantum harmonic oscillator,
\textit{Phys. Rev. A} \textbf{53} 70-108 (1996).

\bibitem{barnett1988} D. T. Pegg and S. M. Barnett, Unitary Phase Operator in Quantum Mechanics, \textit{Europhys. Lett.} \textbf{6} 483-487 (1988).

\bibitem{popov1992} V. N. Popov and V. S. Yarunin, Quantum and Quasi-classical States of the Photon Phase Operator, \textit{J. Mod. Opt.} \textbf{39(7)} 1525-1531 (1992).



\bibitem{barnett2007} S. M. Barnett and J. A. Vaccaro (eds),  \emph{The Quantum Phase Operator: A Review},  Series in Optics and Optoelectronics,  Taylor \& Francis, 2007.



%\bibitem{ayub2011} M. Ayub, K. Naseer, M. Ali, and F.  Saif. (2011). Atom Optics Quantum Pendulum, 1-14 (2011).


\bibitem{busch01} P. Busch, P. Lahti, J.-P. Pellonp\"a\"a, and K. Ylinen, Are numbers and
phase complementary observables? \textit{J. Phys. A: Math. Gen.} \textbf{34}  5923-5935 (2001).


\bibitem{bukiuwer16} P. Busch, J. Kiukas, and  R.F. Werner, Sharp uncertainty relations for number and angle, arXiv:1604.00566v1 [quant-ph]

\bibitem{galapon02}  E.A. Galapon, Pauli's theorem and quantum canonical pairs: the
consistency of a bounded, self-adjoint time operator canonically conjugate to a Hamiltonian
with non-empty point spectrum,  \textit{R. Soc. Lond. Proc. Ser. A  Math. Phys. Eng. Sci. }
\textbf{458}  451-472 (2002); also  in  {\it Time in quantum mechanics} Vol. 2, 25-63, Lecture
Notes in Phys. \textbf{789}, Springer, Berlin (2009).


\bibitem{bergaz13} H. Bergeron and J.-P. Gazeau, Integral
quantisations with two basic examples, \textit{Annals of Physics}, \textbf{344} 43-68 (2014);
 arXiv:1308.2348 {[}quant-ph{]}

\bibitem{aagbook14} S.T. Ali, J.-P. Antoine, and J.-P. Gazeau,  \emph{Coherent States,
Wavelets and their Generalizations\/} 2d edition, Theoretical and Mathematical Physics,
Springer, New York, 2014.	


\bibitem{gabafre14} M. Baldiotti, R. Fresneda, and J.P. Gazeau,  Three examples of covariant
integral quantisation, \textit{Proceedings of the 3d International Satellite Conference on
Mathematical Methods in Physics - ICMP 2013}, \textit{Proceedings of Science}, 03 (2014).

\bibitem{gazszaf16} J.P. Gazeau and  F. H. Szafraniec, Three paths toward the quantum angle operator,    \textit{Annals of Physics (NY)} \textbf{375} 16-35 (2016);  arXiv:1602.07319 [quant-ph]

\bibitem{argaho12} I. Aremua, J. P. Gazeau and M. N. Hounkonnou, Action-angle coherent states
for quantum systems with cylindric phase space, \textit{J. Phys. A: Math. Theor.} \textbf{45}
335302-1-16 (2012).

\bibitem{pedro2007} P. L. Garcia de Leon and J. P. Gazeau, Coherent state quantisation and phase operator, \textit{Phys. Lett. A} \textbf{361(4)} 301-304 (2007).


\bibitem{debievre89} S. De Bi\`evre, Coherent states over symplectic homogeneous spaces,
\textit{J. Math. Phys.} \textbf{30}  1401-1407 (1989).

\bibitem{bowo74} C.P. Boyer and K.B. Wolf,  III, Configuration and phase descriptions of quantum systems possessing an sl$(2,\R)$ dynamical algebra,  \textit{J. Math. Phys.} \textbf{16}  1493-1502 (1975).

\bibitem{isham84} C.J. Isham, in \textit{Relativity, Groups and Topology II}, edited by B.
S. DeWitt and R. Stora, \textit{Proceedings of the Les Houches Summer
School of Theoretical Physics, XL}, 1983,  Elsevier, Amsterdam,
1984, pp. 1059-1290.

\bibitem{niatchwo98} L. M. Nieto, N. M. Atakishiyev, S. M. Chumakov and K. B. Wolf, Wigner distribution function for Euclidean systems, \textit{J. Phys. A} \textbf{31(16)} 3875-3895 (1989). 


\bibitem{kastrup06} H.A. Kastrup, quantisation of the canonically conjugate pair angle and orbital angular momentum, \textit{Phys. Rev. A} \textbf{73} 052104-1-26 (2006).

%\bibitem{berberian}  S. K. Berberian, Borel Spaces, Functional analysis and its applications (Nice, 1986), 134-197, \textit{ICPAM Lecture Notes}, World Sci. Publishing, Singapore, 1988, available at  \textcolor {hyptxt}{\href{https://www.ma.utexas.edu/mp\_arc/c/02/02-156.pdf} {preprint U Austin}}
\bibitem{main:ch5:debgo} S. De Bi\`evre and J.A. Gonz\'alez,  Semiclassical behaviour of coherent states on the circle, \newblock In A. Odzijewicz et al, editors, {\em Quantisation and Coherent States Methods in Physics} \newblock Singapore: World Scientific, 1993.

\bibitem{main:ch5:kopap} K. Kowalski, J. Rembielinski and L.C. Papaloucas, Coherent states for a quantum particle on a circle, \textit{J. Phys. A: Math. Gen. }  \textbf{29} 4149-4167 (1996).

\bibitem{main:ch5:delgo} J.A. Gonz\'alez and M.A. del Olmo, Coherent states on the circle, \textit{J. Phys. A: Math. Gen.} \textbf{31} 8841-8857 (1998).

\bibitem{main:ch5:kowrem1}  K. Kowalski and  J. Rembielinski, Exotic behaviour of a quantum particle on a circle, \textit{Phys. Lett. A} \textbf{293} 109-115 (2002).

\bibitem{main:ch5:hallmitch} B.C. Hall and J.J. Mitchell,  Coherent states on spheres, \textit{J. Math. Phys.} \textbf{43} 1211-1236 (2002).

\bibitem{main:ch5:kowrem2} K. Kowalski and  J. Rembielinski, On the uncertainty relations and squeezed states for the quantum mechanics on a circle, \textit{J. Phys. A: Math. Gen. } \textbf{36} 1405-1414 (2002).

\bibitem{trifonov03} D. A. Trifonov, Comment on ``On the uncertainty relations and squeezed states for the quantum mechanics on a circle'', \textit{J. Phys. A: Math. Gen. } \textbf{35} 2197-2202 (2003).

\bibitem{main:ch5:kowrem3} K. Kowalski and  J. Rembielinski, Reply to the ``Comment on ``On the uncertainty relations and squeezed states for the quantum mechanics on a circle'', \textit{J. Phys. A: Math. Gen. }  \textbf{36} 5695-5698 (2003).

\bibitem{gilmore72} R. Gilmore,  Geometry of symmetrized states,   \textit{Ann. Phys. (NY)}
\textbf{74} 391-463 (1972);  On properties of coherent states,
           \textit{Rev. Mex. Fis.} \textbf{23} 143-187  (1974).
           
\bibitem{perelomov86} A.M. Perelomov, \textit{ Generalized Coherent States and their
            Applications\/}, Springer-Verlag, Berlin, 1986.

\bibitem{lieb73} E.H. Lieb, The classical limit of quantum spin systems,
 \textit{Commun. Math. Phys.} \textbf{31}  327-340 (1973).


\bibitem{berezin74} F.A. Berezin,  quantisation,   \textit{Math. USSR Izvestija} \textbf{8}  1109-1165 (1974);
       General concept of quantisation, \textit{Commun. Math. Phys.} \textbf{40}  153-174 (1975).
       
 
\bibitem{schilling06} R.L. Schilling, Measures, Integrals and Martingales,
 Cambridge University Press, 2006.
%\bibitem{lanztho10} I. Urizar-Lanz and G. T\'oth, Number-operator-annihilation-operator
%uncertainty as an alternative for the number-phase uncertainty relation,  \textit{Phys.
%Rev. A} \textbf{81}  052108-1-7 (2012).

%\bibitem{cahillglauber69}
%K.E. Cahill and R. Glauber, Ordered expansion in Boson Amplitude Operators,
%\textit{Phys. Rev.} {\bf 177} (1969) 1857-1881;  Density Operators and Quasiprobability
%Distributions,   \textit{Phys. Rev.} {\bf 177}  1882-1902 (1969).

%\bibitem{AgaWo70} B.~S. Agarwal and E. Wolf, Calculus for Functions of Noncommuting Operators
%and General Phase-Space Methods in Quantum Mechanics, Phys. Rev. D \textbf{2} 2161; 2187;
%2206 (1970).



\bibitem{gazbook09}
J.P. Gazeau  {\it Coherent States in Quantum Physics} (Berlin: Wiley-VCH) 2009.

%\bibitem{magnus66}   Wilhelm Magnus, Fritz Oberhettinger, and Raj~Pal  Soni.
%\newblock {\em Formulas and Theorems for
%the Special Functions of Mathematical Physics}
 %\newblock Springer-Verlag,  Berlin, Heidelberg and New York, 1966

%\bibitem{gazolm13} J.P. Gazeau and  M. del Olmo, $q$-coherent states quantisation of the
%harmonic oscillator,    \textit{Annals of Physics (NY)} \textbf{330} 220-245
%(2012);  arXiv:1207.1200 [quant-ph]

\end{thebibliography}
%%%%%%%%%%%%%%%%%%%%%%%%%%%%%%%%%%%%%%%%%%%%%%%%%%%%%%%%%%%%%%%%%%%%%%%%%%%%%

%---------------------------------------------------------------------------------------
\end{document}